\documentclass[a4paper,onecolumn,10pt,accepted=2024-04-19,shorttitle=papers,volume=6,issue=2]{compositionalityarticle}
\pdfoutput=1
\usepackage[utf8]{inputenc}
\usepackage[T1]{fontenc}
\usepackage{microtype}
\usepackage[shortlabels]{enumitem}
\usepackage{booktabs}
\usepackage[font=small]{caption}
\usepackage{subcaption}

\usepackage{amsmath,amssymb,amsthm,mathtools}
\newcommand{\define}[1]{\textbf{#1}}
\newcommand{\Z}{\mathbb{Z}}
\newcommand{\R}{\mathbb{R}}
\newcommand{\To}{\Rightarrow}
\newcommand{\monicto}{\rightarrowtail}
\newcommand{\longdashv}{\mathrel{\relbar\mkern-4mu\relbar\mkern-9mu\vcenter{\hbox{$\dashv$}}}}
\newcommand{\Sgn}{\mathsf{Sgn}}
\DeclareMathOperator{\sgn}{sgn}
\DeclareMathOperator{\negate}{neg}
\DeclareMathOperator{\one}{one}
\DeclareMathOperator{\ev}{ev}
\DeclareMathOperator{\src}{src}
\DeclareMathOperator{\tgt}{tgt}
\DeclareMathOperator{\diag}{diag}
\newcommand{\cat}[1]{\mathsf{#1}}
\newcommand{\CAT}[1]{\mathsf{#1}}
\newcommand{\Sch}[1]{\mathsf{Sch}(#1)}
\newcommand{\Set}{\CAT{Set}}
\newcommand{\FinSet}{\CAT{FinSet}}
\newcommand{\Mod}[1]{\CAT{Mod}_{#1}}
\newcommand{\Vect}{\CAT{Vect}_{\R}}
\newcommand{\Poly}{\CAT{Poly}_{\R}}
\newcommand{\Con}{\CAT{Con}}
\newcommand{\Graph}{\CAT{Graph}}
\newcommand{\FinGraph}{\CAT{FinGraph}}
\newcommand{\Cat}{\CAT{Cat}}
\newcommand{\SgnGraph}{\CAT{SgnGraph}}
\newcommand{\FinSgnGraph}{\CAT{FinSgnGraph}}
\newcommand{\SgnCat}{\CAT{SgnCat}}

\newcommand{\PetriLink}{\CAT{LPetri}}
\newcommand{\SgnPetri}{\CAT{SgnPetri}}
\newcommand{\LinParaDynam}{\CAT{Para}(\Dynam)}
\newcommand{\ConParaDynam}{\CAT{Para}(\Dynam_+)}
\newcommand{\op}{\mathrm{op}}
\DeclareMathOperator{\id}{id}
\DeclareMathOperator{\Ob}{Ob}
\DeclareMathOperator{\Hom}{Hom}
\DeclareMathOperator{\Disc}{Disc}
\DeclareMathOperator{\Path}{Path}
\DeclareMathOperator{\Int}{Int}
\DeclareMathOperator{\Dynam}{Dynam}
\DeclareMathOperator{\LV}{LV}
\newcommand{\Open}[1]{\mathbb{O}\mathsf{pen}(#1)}
\newcommand{\Csp}{\mathbb{C}\mathsf{sp}}

\usepackage[breaklinks,bookmarksnumbered]{hyperref}
\usepackage[capitalize,noabbrev,nameinlink]{cleveref}

\usepackage[style=alphabetic,maxbibnames=10]{biblatex}
\renewbibmacro{in:}{}

\usepackage{thmtools}
\declaretheorem[within=section]{theorem}
\declaretheorem[sibling=theorem]{proposition}
\declaretheorem[sibling=theorem]{lemma}
\declaretheorem[sibling=theorem,style=definition]{definition}
\declaretheorem[sibling=theorem,style=remark]{remark}
\declaretheorem[sibling=theorem,style=remark]{example}

\usepackage{tikz}
\usepackage{quiver}
\usetikzlibrary{cd, fit}
\tikzstyle{species}=[circle,draw=blue!75,fill=blue!10,minimum size=5mm]
\tikzstyle{transition}=[rectangle,draw=black!75,fill=black!10,
  minimum width=1mm, minimum height=5mm]
\tikzstyle{arc}=[->,>=stealth]
\tikzstyle{link}=[->,>=stealth,dashed]

\title{A compositional account of motifs, mechanisms, and dynamics in
  biochemical regulatory networks}
\date{}
\author{Rebekah Aduddell}
\affiliation{University of Texas at Arlington, Mathematics Department, Arlington, TX 76019, USA}
\email{rjaduddell@gmail.com}
\author{James P.\ Fairbanks}
\affiliation{University of Florida, Computer \& Information Science \& Engineering, Gainesville, FL 32611, USA}
\email{fairbanksj@ufl.edu}
\author{Amit Kumar}
\affiliation{Louisiana State University, Mathematics Department, Baton Rouge, LA 70803, USA}
\email{akuma25@lsu.edu}
\author{Pablo S.\ Ocal}
\affiliation{University of California, Los Angeles (UCLA), Mathematics Department, Los Angeles, CA 90095, USA}
\email{socal@math.ucla.edu}
\author{Evan Patterson}
\affiliation{Topos Institute, Berkeley, CA 94704, USA}
\email{evan@epatters.org}
\author{Brandon T.\ Shapiro}
\affiliation{University of Virginia, Mathematics Department, Charlottesville, VA 22904, USA}
\email{brandonshapiro@virginia.edu}

\newcommand{\authorsforheader}{}
\newcommand{\receiveddate}{2023-03-28}

\begin{document}
\maketitle

\begin{abstract}
  Regulatory networks depict promoting or inhibiting interactions between
  molecules in a biochemical system. We introduce a category-theoretic formalism
  for regulatory networks, using signed graphs to model the networks and signed
  functors to describe occurrences of one network in another, especially
  occurrences of network motifs. With this foundation, we establish functorial
  mappings between regulatory networks and other mathematical models in
  biochemistry. We construct a functor from reaction networks, modeled as Petri
  nets with signed links, to regulatory networks, enabling us to precisely
  define when a reaction network could be a physical mechanism underlying a
  regulatory network. Turning to quantitative models, we associate a regulatory
  network with a Lotka-Volterra system of differential equations, defining a
  functor from the category of signed graphs to a category of parameterized
  dynamical systems. We extend this result from closed to open systems,
  demonstrating that Lotka-Volterra dynamics respects not only inclusions and
  collapsings of regulatory networks, but also the process of building up
  complex regulatory networks by gluing together simpler pieces. Formally, we
  use the theory of structured cospans to produce a lax double functor from the
  double category of open signed graphs to that of open parameterized dynamical
  systems. Throughout the paper, we ground the categorical formalism in examples
  inspired by systems biology.
\end{abstract}

\section{Introduction}

The genes, proteins, and RNA molecules that comprise living cells interact in
complex, varied ways to sustain the cell throughout its lifecycle and respond to
changes in its environment. Intensive experimental study of these interactions
is distilled in an idealized form as \emph{regulatory networks}, a kind of
directed graph in which vertices represent molecules and edges represent
interactions between molecules (\cref{fig:first-regnet}). The edges are labeled
with a positive or negative sign according to whether the interaction is
activating or inhibiting. Regulatory networks are the subject of a large body of
experimental and theoretical work, notably reviewed by Alon
\cite{alon2007,alon2019} and Tyson et al. \cite{tyson2010,tyson2019} among
others. Particular attention has been paid to \emph{network motifs}
\cite{alon2007,tyson2010}, the simple but functionally meaningful patterns that
recur frequently in regulatory networks, and to various quantitative
\emph{dynamics} \cite{tyson2019} that can be assigned to the networks.

\begin{figure}[h]
  \begin{equation*}
    \begin{tikzcd}[column sep=small]
      {} & {\mathrm{Ash1}} & {}\\
      {\mathrm{Cdk1/ClbS}} & {} & {\mathrm{Sld2}}
      \arrow[from=1-2, to=2-3]
      \arrow[maps to, no head, from=1-2, to=2-1]
      \arrow[curve={height=10pt}, from=2-1, to=2-3]
      \arrow[curve={height=10pt}, maps to, no head, from=2-3, to=2-1]
    \end{tikzcd}
  \end{equation*}
  \caption{A small biochemical regulatory network: regulation of Sld2 by Cdk1 or
    ClbS with Ash1 as a predicted transcription factor. Adapted from
    Csik{\'a}sz-Nagy et al. \cite[Figure 3C]{csikasz-nagy2009}.}
  \label{fig:first-regnet}
\end{figure}
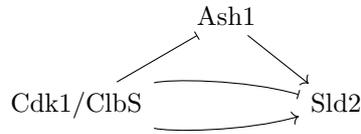

Although regulatory networks are simple enough to define mathematically---we
shall define them to be directed graphs, possibly with multiple edges and loops,
whose edges are assigned a positive or negative sign---important scientific
concepts involving them, such as occurrences of motifs in networks and
biochemical mechanisms generating networks, are often treated imprecisely.
Likewise for relationships between regulatory networks and other mathematical
models in biochemistry, particularly dynamical models based on ordinary or
stochastic differential equations. Hence a first aim of this paper is to put
certain concepts and relations concerning regulatory networks on a firm
mathematical footing. To do so, we will use methods from category theory.

Category theory, in both the small and the large, is a natural tool for this
study. In saying that a motif \emph{occurs} in a network, one should allow for
the possibility that the occurrence is indirect, involving a sequence of
appropriately signed interactions. For example, positive autoregulation can
occur directly but also indirectly through a double-negative feedback loop.
Since a small category is nothing other than a graph in which consecutive edges
can be composed, subject to certain laws, regulatory networks should be viewed
not only as signed graphs (\cref{sec:sgn-graph}) but also as signed categories
freely generated by those (\cref{sec:sgn-cat}). Sign-preserving functors, unlike
sign-preserving graph homomorphisms, can express indirect occurrences and are in
this sense a better notion of morphism for regulatory networks. Here we are
doing category theory \emph{in the small}, using categories as algebraic
structures comparable to familiar ones like graphs, groups, and monoids.

Having laid these foundations for regulatory networks, we turn to category
theory \emph{in the large}, a mathematical theory of structure well suited to
describe the passages between regulatory networks and other mathematical models
of biochemical systems. Formally speaking, these passages are functors into or
out of the category of regulatory networks. Making a functor is significantly
stronger than making an objects-only mapping, as is typically done in the
literature, since if morphisms of signed graphs formalize relationships between
different regulatory networks, then functorality requires that these
relationships be transported to or from other models of interest. By contrast,
an objects-only mapping is, abstractly speaking, entirely unconstrained and so
is capable of acting highly irregularly across different models of a given
class. Functorality thus serves as a kind of safeguard for model transformation:
it does not, on its own, ensure that a transformation makes good scientific
sense but it does impose nontrivial logical constraints and coherences.

A first illustration of this principle is the connection between regulatory
networks and biochemical reaction networks (\cref{sec:sgn-petri}). When modeling
the complex biochemical systems that constitute a living cell, it is often
practically necessary to abstract away certain details of the underlying
chemical processes. Regulatory networks generally do not capture all the species
or reactions involved in a given system; nor can they capture multispecies
reactions faithfully because they describe only pairwise interactions. Given
that regulatory networks are, to some degree, phenomenological models, it is
natural to ask whether a given network could arise as a summary of a specific
chemical process. The latter are described by \emph{biochemical reaction
  networks}, graph-like structures allowing reactions or transitions with
multiple inputs and outputs. Inspired by graphical syntax from systems biology
\cite{voit2000,voit2013}, we formalize reaction networks as ``Petri nets with
links,'' and we construct a functor from the category of Petri nets with signed
links to the category of signed graphs. This functor enables us to propose a
formal definition for when a reaction network could be a \emph{mechanism} for a
regulatory network, a concept that is rarely if ever treated in a precise way.

This concludes the content of \cref{sec:qualitative}. In
\cref{sec:quantitative}, we turn from qualitative to quantitative analysis,
seeking a functorial assignment of continuous dynamics to regulatory networks.
Although rarely made explicitly functorial, systematic ways to formulate a model
belonging to a mathematically homogeneous class of models are ubiquitous in
science. Voit calls these ``canonical representations'' or ``canonical
models,''\footnote{``Canonical'' models in systems biology should not be
  confused with the unrelated, in fact incompatible, notion of ``canonical''
  constructions in category theory.} and identifies Lotka-Volterra models and
BST models/S-systems as two prominent examples in biology \cite[\S 3]{voit2013}.
Reflecting their phenomenological status, regulatory networks do not admit a
single, obvious dynamical interpretation, and so a wide variety of dynamical
models have been considered, spanning the discrete and continuous, deterministic
and stochastic \cite{tyson2019}. We consider Lotka-Volterra systems of ordinary
differential equations. While not necessarily the most biologically plausible,
Lotka-Volterra systems are among the simplest possible continuous models and so
are a natural place to begin a functorial study.

A \emph{Lotka-Volterra system} of equations has the form
\begin{equation*}
  \dot x_i = \rho_i\, x_i + \sum_{j=1}^n \beta_{i,j}\, x_i\, x_j,
  \qquad i = 1,\dots,n.
\end{equation*}
or equivalently, has logarithmic derivatives that are affine functions of the
state variables:
\begin{equation*}
  \frac{d}{dt}[\log x_i(t)] = \rho_i + \sum_{j=1}^n \beta_{i,j}\, x_j(t),
  \qquad i = 1,\dots,n.
\end{equation*}
The coefficients $\rho_i$ specify baseline rates of growth or decay, according
to their sign, and the coefficients $\beta_{i,j}$ rates of activation or
inhibition, according to their sign. We construct a functor that sends a signed
graph (regulatory network) to a Lotka-Volterra model that constrains the signs
of the rate coefficients (\cref{sec:lv}). As a prerequisite, we define a
category of parameterized dynamical systems (\cref{sec:para-dynam}), a
construction of intrinsic interest that is by no means restricted to
Lotka-Volterra dynamics. By working with signed graphs, rather than merely
graphs, we ensure that scientific knowledge about whether interactions are
promoting or inhibiting is reflected in both the syntax and the quantitative
semantics.

In order to comprehend complex biological systems, we must decompose them into
small, readily understandable pieces and then compose them back together to
reproduce the behavior of the original system. This is the mantra of systems
biology, which stresses that compositionality is no less important than
reductionism in biology. With this motivation, a secondary aim of this paper is
to extend the above constructions from closed systems to open ones, which can be
composed together by gluing them along their interfaces. Mathematically, we pass
from categories to \emph{double categories}\footnote{Early work on categorical
  systems theory, including on structured cospans \cite{fiadeiro2007}, was based
  on \emph{bicategories}. For technical reasons explained in
  \cite{patterson2023}, it is increasingly common to use double categories
  instead, as in recent work on structured cospans \cite{baez2020}. This entails
  no loss since every double category has an underlying bicategory.}
\cite{grandis2019}, two-dimensional categorical structures in which the usual
morphisms of systems compose along one direction (by convention, the
``vertical'' one) and open systems compose along the other direction (the
``horizontal'' one). The double categories of open systems are further equipped
with \emph{monoidal products}, enabling systems to be composed not just in
sequence but also in parallel. Among other results, we show that the
Lotka-Volterra dynamics functor extends to a monoidal lax double functor from
the monoidal double category of open signed graphs to that of open parameterized
dynamical systems (\cref{sec:open-lv}).

The mathematics developed here is motivated by biochemistry but need not be
restricted to it. Famously, Lotka-Volterra systems originated in ecology to
model predator-prey dynamics \cite{lotka1925}. Regulatory networks and
Lotka-Volterra systems can be used as generic models of entities that
``regulate'' each other in some manner, be it at the scale of individual cells
or animal ecosystems. Regulatory networks are highly reminiscent of the
\emph{causal loop diagrams} in system dynamics \cite[Chapter 5]{sterman2000},
where the latter explicitly label feedback loops and their polarities.

The language of category theory is indispensable to this work but the level of
knowledge assumed of the reader is not constant. We assume throughout that the
reader is familiar with the basic notions of category theory, such as
categories, functors, and natural transformations. Our main reference for facts
about category theory is Riehl's text \cite{riehl2016}, although there are many
others. In the definitions and theorem statements, we have tried to minimize the
technical level and explicate the ideas in concrete terms. In the proofs, we
have aimed for efficiency and freely use concepts and results from the
literature that do not appear in the main text. The reader can omit the proofs
without disrupting the continuity of the paper.

\paragraph{Acknowledgments.}

The authors thank the American Mathematical Society (AMS) for hosting the 2022
Mathematical Research Community (MRC) on Applied Category Theory, where this
research project began. The AMS MRC was supported by NSF grant 1916439. We thank
John Baez, our group's mentor at the MRC, for suggesting this project and for
much helpful advice along the way. Authors Fairbanks, Patterson, and Shapiro
acknowledge subsequent support from the DARPA ASKEM and Young Faculty Award
programs through grants HR00112220038 and W911NF2110323. Author Ocal
acknowledges subsequent support from an AMS-Simons Travel Grant and from the
Hausdorff Research Institute for Mathematics funded by the German Research
Foundation (DFG) under Germany's Excellence Strategy - EXC-2047/1 - 390685813.

\section{Qualitative analysis: motifs and mechanisms}
\label{sec:qualitative}

\subsection{Regulatory networks as signed graphs}
\label{sec:sgn-graph}

To begin, we clarify the notion of graph to be used throughout in this paper.
The following definition is standard among category theorists. In other fields,
it might be called a ``directed multigraph,'' but we will call it simply a
``graph.''

\begin{definition}[Graphs] \label{def:graph}
  The \define{schema for graphs} is the category $\Sch{\Graph}$ freely generated
  by two parallel morphisms:
  \begin{equation*}
    \begin{tikzcd}
      V & E
      \arrow["{\mathrm{src}}"', shift right=1, from=1-2, to=1-1]
      \arrow["{\mathrm{tgt}}", shift left=1, from=1-2, to=1-1]
    \end{tikzcd}.
  \end{equation*}
  A \define{graph} is a functor $X: \Sch{\Graph} \to \Set$, also known as a
  \emph{copresheaf} on $\Sch{\Graph}$.\footnote{Applied category theorists often
    take set-valued functors to be covariant (i.e., as copresheaves) rather than
    contravariant (i.e., as the more traditional presheaves), for reasons of
    convenience visible in works such as \cite{spivak2021}.} A \define{graph
    homomorphism} from a graph $X$ to another graph $Y$ is a natural
  transformation $\phi: X \to Y$. Graphs and graph homomorphisms form the category
  $\Graph$.
\end{definition}

To restate the definition in explicit terms, a graph $X$ consists of
\begin{itemize}[noitemsep]
  \item a set $X(V)$ of \define{vertices};
  \item a set $X(E)$ of \define{edges}; and
  \item functions $X(\src), X(\tgt): X(E) \to X(V)$, assigning to each edge its
    \define{source} and \define{target}.
\end{itemize}
A graph homomorphism $\phi: X \to Y$ consists of a function
$\phi_V: X(V) \to Y(V)$, the \define{vertex map}, and another function
$\phi_E: X(E) \to Y(E)$, the \define{edge map}. These maps must preserve sources
and targets, meaning that the following squares commute:
\begin{equation*}
  \begin{tikzcd}
    {X(E)} & {X(V)} \\
    {Y(E)} & {Y(V)}
    \arrow["{X(\mathrm{src})}", from=1-1, to=1-2]
    \arrow["{\phi_E}"', from=1-1, to=2-1]
    \arrow["{Y(\mathrm{src})}"', from=2-1, to=2-2]
    \arrow["{\phi_V}", from=1-2, to=2-2]
  \end{tikzcd}
  \qquad\qquad
  \begin{tikzcd}
    {X(E)} & {X(V)} \\
    {Y(E)} & {Y(V)}
    \arrow["{X(\mathrm{tgt})}", from=1-1, to=1-2]
    \arrow["{\phi_E}"', from=1-1, to=2-1]
    \arrow["{Y(\mathrm{tgt})}"', from=2-1, to=2-2]
    \arrow["{\phi_V}", from=1-2, to=2-2]
  \end{tikzcd}.
\end{equation*}

We now turn to the main notion of this section, signed graph. Write $\Sgn$ for
the set of (nonzero) signs, whose two elements may be denoted $\{1,-1\}$ or
$\{+,-\}$. The set of signs is an abelian group, isomorphic to the cyclic group
$\Z_2$, under the usual multiplication.

\begin{definition}[Signed graphs] \label{def:sgn-graph}
  The \define{category of signed graphs} is the slice category
  \begin{equation*}
    \SgnGraph \coloneqq \Graph/\Sgn,
  \end{equation*}
  where, by abuse of notation, $\Sgn$ is regarded as a graph with one vertex and
  two loops.
\end{definition}

Unpacking the definition, a \define{signed graph} is seen to be a graph $X$
equipped with a function $X(\sgn): X(E) \to \Sgn$ that assigns a sign to each
edge. Given signed graphs $X$ and $Y$, a \define{morphism of signed graphs} from
$X$ to $Y$ is a graph homomorphism $\phi$ that preserves signs, meaning that the
following triangle commutes:
\begin{equation*}
  \begin{tikzcd}[column sep=small]
    {X(E)} && {Y(E)} \\
    & {\mathsf{Sgn}}
    \arrow["{X(\mathrm{sgn})}"', from=1-1, to=2-2]
    \arrow["{Y(\mathrm{sgn})}", from=1-3, to=2-2]
    \arrow["{\phi_E}", from=1-1, to=1-3]
  \end{tikzcd}.
\end{equation*}

Signed graphs are a mathematical description of the regulatory networks studied
in systems biology \cite{alon2007,tyson2010}. For the purposes of this paper, we
will simply define a \define{regulatory network} to be a signed graph. The
vertices of the graph represent the components of the network, which could be
proteins, genes, or RNA molecules. Signed edges represent interactions between
components, where the source has the effect of either
\emph{activating}/\emph{promoting} the target (positive sign) or
\emph{inhibiting}/\emph{repressing} it (negative sign). As is customary, we
denote activation interactions by arrows with pointed heads ($\longrightarrow$)
and inhibition interactions by arrows with flat heads ($\longdashv$). For
instance, the two drawings
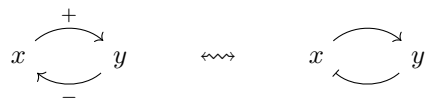
\begin{equation*}
  \begin{tikzcd}
    x & y
    \arrow["{+}", curve={height=-12pt}, from=1-1, to=1-2]
    \arrow["{-}", curve={height=-12pt}, from=1-2, to=1-1]
  \end{tikzcd}
  \qquad\leftrightsquigarrow\qquad
  \begin{tikzcd}
    x & y
    \arrow[curve={height=-12pt}, from=1-1, to=1-2]
    \arrow[curve={height=12pt}, maps to, no head, from=1-1, to=1-2]
  \end{tikzcd}
\end{equation*}
represent the same network, a negative feedback loop in which $x$ activates $y$,
which in turn inhibits $x$ \cite[Scheme 1, Motif B]{tyson2010}.

In the literature \cite{tyson2010}, regulatory networks are often represented
mathematically as sign-valued matrices. This approach is a special case of ours
in that an $n$-by-$n$ matrix valued in $\{+1,-1,0\}$ can be interpreted as a
simple signed graph on $n$ vertices, with signed edges defined by the nonzero
matrix elements. Unlike the matricial formalism, our formalism allows multiple
edges between the same pair of edges, which can model multiple interactions
based on different mechanisms. Allowing multiple edges and self-loops also
ensures that graphs and signed graphs form well behaved categories, as the
following proposition shows.

\begin{proposition}
  The category of signed graphs is complete (has all limits) and cocomplete (has
  all colimits).
\end{proposition}
\begin{proof}
  Because $\Graph$ is a copresheaf category, it is complete and cocomplete
  \cite[Proposition 3.3.9]{riehl2016}. The slice category
  $\SgnGraph = \Graph/\Sgn$ is hence also complete and cocomplete
  \cite[Proposition 3.5.5]{riehl2016}; alternatively, this follows because
  slices of copresheaf categories are again (equivalent to) copresheaf
  categories \cite[Remark p.\ 303]{street2000}.
\end{proof}

A morphism of signed graphs can do two things. Most obviously, it can pick out a
signed graph as a subobject of another one, via a sign-preserving subgraph
embedding. A signed graph morphism can also collapse multiple vertices onto a
single vertex, and multiple edges onto a single edge with the same sign, in the
restrictive sense permitted by a graph homomorphism. To illustrate, consider the
following morphism inspired by Alon's review \cite[Figure 5]{alon2007}.
\begin{equation} \label{eq:arginine-morphism}
  \begin{tikzcd}[column sep=small]
    && {\mathrm{argR}} \arrow[loop above, maps to, no head] \\
    {\mathrm{argCBH}} & {\mathrm{argD}} & {\mathrm{argE}} & {\mathrm{argF}} & {\mathrm{argI}}
    \arrow[maps to, no head, from=2-1, to=1-3]
    \arrow[maps to, no head, from=2-2, to=1-3]
    \arrow[maps to, no head, from=2-3, to=1-3]
    \arrow[maps to, no head, from=2-4, to=1-3]
    \arrow[maps to, no head, from=2-5, to=1-3]
  \end{tikzcd}
  \qquad\longrightarrow\qquad
  \begin{tikzcd}
    {\mathrm{argR}} \arrow[loop above, maps to, no head] \\
    {\mathrm{arg*}}
    \arrow[maps to, no head, from=2-1, to=1-1]
  \end{tikzcd}
\end{equation}
The network in the domain is a ``single-input module'' in the arginine
biosynthesis system, in which the regulator argR represses five different
enzymes (argCHB, argD, etc.)\ involved in producing arginine. The morphism above
forgets the distinction between these enzymes, collapsing them into a catch-all
entity labeled ``$\mathrm{arg*}$''. These two functions---embedding and
collapsing---are \emph{all} that a signed graph morphism can do. More precisely,
any morphism of signed graphs factors essentially uniquely as an epimorphism
(morphism with surjective vertex and edge maps) followed by a monomorphism
(morphism with injective vertex and edge maps), using the epi-mono factorization
available in any copresheaf category, or more generally in any topos \cite[\S
IV.6]{maclane1994}. In the next section we will consider a more flexible notion
of morphism between signed graphs.

Colimits of signed graphs can be used to construct a category, or rather a
double category, of \emph{open} signed graphs. Composition of open signed graphs
formalizes the process of building large regulatory networks from smaller
pieces, including network motifs.

\begin{proposition}[Open signed graphs] \label{prop:open-sgn-graphs}
  There is a symmetric monoidal double category of open signed graphs,
  $\Open{\SgnGraph}$, having
  \begin{itemize}[noitemsep]
    \item as objects, sets $A, B, C, \dots$;
    \item as vertical morphisms, functions $f: A \to B$;
    \item as horizontal morphisms, \define{open signed graphs}, which consist of a
      signed graph $X$ together with a cospan of sets
      $A_0 \xrightarrow{\ell_0} X(V) \xleftarrow{\ell_1} A_1$;
    \item as cells, \define{morphisms of open signed graphs}
      $(X,\ell_0,\ell_1) \to (Y,m_0,m_1)$, which consist of a map of signed
      graphs $\phi: X \to Y$ along with functions $f_i: A_i \to B_i$, $i=0,1$,
      making the following diagram commute:
      \begin{equation*}
        \begin{tikzcd}
          {A_0} & {X(V)} & {A_1} \\
          {B_0} & {Y(V)} & {B_1}
          \arrow["{\ell_0}", from=1-1, to=1-2]
          \arrow["{f_0}"', from=1-1, to=2-1]
          \arrow["{\ell_1}"', from=1-3, to=1-2]
          \arrow["{\phi_V}"', from=1-2, to=2-2]
          \arrow["{m_0}"', from=2-1, to=2-2]
          \arrow["{f_1}", from=1-3, to=2-3]
          \arrow["{m_1}", from=2-3, to=2-2]
        \end{tikzcd}.
      \end{equation*}
  \end{itemize}
  Vertical composition is by composition in $\Set$ and in $\SgnGraph$.
  Horizontal composition and monoidal products are by pushouts and coproducts in
  $\SgnGraph$, respectively, viewing the sets in the feet of the cospans as
  discrete signed graphs.
\end{proposition}
\begin{proof}
  To construct this symmetric monoidal double category, we use the method of
  structured cospans \cite{fiadeiro2007} in its double-categorical form
  \cite{baez2020}. The categories of sets and of signed graphs are related by an
  adjoint pair of functors
  \begin{equation*}
    \begin{tikzcd}
      {\mathsf{Set}} & {\mathsf{SgnGraph}}
      \arrow[""{name=0, anchor=center, inner sep=0}, "{\mathrm{Disc}}", curve={height=-12pt}, from=1-1, to=1-2]
      \arrow[""{name=1, anchor=center, inner sep=0}, "{\mathrm{ev}_V}", curve={height=-12pt}, from=1-2, to=1-1]
      \arrow["\dashv"{anchor=center, rotate=-90}, draw=none, from=0, to=1]
    \end{tikzcd}.
  \end{equation*}
  Here $\ev_V: \SgnGraph \to \Set$ is the \define{evaluation at $V$} functor,
  sending a signed graph $X$ to its set of vertices $X(V)$ and a morphism of
  signed graphs $\phi$ to its vertex map $\phi_V$, and
  $\Disc: \Set \to \SgnGraph$ is the \define{discrete signed graph} functor,
  sending a set $A$ to the signed graph with vertex set $A$ and no edges. We
  obtain a symmetric monoidal double category of open signed graphs as the
  $L$-structured cospans for the functor $L \coloneqq \Disc: \Set \to \SgnGraph$
  \cite[Theorems 2.3 and 3.9]{baez2020}.

  To show that this symmetric monoidal double category is the same one in the
  proposition statement, suppose that $L \dashv R: \cat{A} \to \cat{X}$ is an
  adjoint pair of functors, where in our application $L = \Disc$ and
  $R = \ev_V$. By the defining bijection of an adjunction, $L$-structured
  cospans, i.e., objects $A_0$ and $A_1$ in $\cat{A}$ together with a cospan
  $L(A_0) \rightarrow X \leftarrow L(A_1)$ in $\cat{X}$, correspond exactly to
  ``$R$-decorated cospans,'' i.e., an object $X$ in $\cat{X}$ together with a
  cospan $A_0 \rightarrow R(X) \leftarrow A_1$ in $\cat{A}$. Furthermore, by the
  naturality of this bijection \cite[Lemma 4.1.3]{riehl2016}, morphisms of
  $L$-structured and $R$-decorated cospans
  \begin{equation*}
    \begin{tikzcd}[column sep=small]
      {L(A_0)} & X & {L(A_1)} \\
      {L(B_0)} & Y & {L(B_1)}
      \arrow[from=1-1, to=1-2]
      \arrow["{L(f_0)}"', from=1-1, to=2-1]
      \arrow[from=1-3, to=1-2]
      \arrow["\phi"', from=1-2, to=2-2]
      \arrow[from=2-1, to=2-2]
      \arrow["{L(f_1)}", from=1-3, to=2-3]
      \arrow[from=2-3, to=2-2]
    \end{tikzcd}
    \qquad\leftrightsquigarrow\qquad
    \begin{tikzcd}[column sep=small]
      {A_0} & {R(X)} & {A_1} \\
      {B_0} & {R(Y)} & {B_1}
      \arrow[from=1-1, to=1-2]
      \arrow["{f_0}"', from=1-1, to=2-1]
      \arrow[from=1-3, to=1-2]
      \arrow["{R(\phi)}"', from=1-2, to=2-2]
      \arrow[from=2-1, to=2-2]
      \arrow["{f_1}", from=1-3, to=2-3]
      \arrow[from=2-3, to=2-2]
    \end{tikzcd}
  \end{equation*}
  related by the adjunction are equivalent in that one diagram commutes if and
  only if the other does. We will tacitly reuse this reasoning in future
  constructions, such as \cref{prop:open-sgn-cats} below.
\end{proof}

Let us illustrate compositionality as a means of building larger regulatory
networks from smaller ones. The following example is adapted from Keurentjes et
al.\ \cite[Figure 1]{keurentjes2011}, later reproduced in \cite[Figure
1.7]{voit2018}.

\begin{example}[Stress response system in plants]
  When a plant perceives stress factors, it has three main biotic responses,
  called salicylic acid (SA) signaling, jasmonic acid (JA) signaling, and
  ethylene (ET) signaling. These processes promote the transcription factors
  WRKY, MYC2, and ERF, respectively, which in turn activate the genes
  responsible for responding to SA, JA, and ET. As it appears in \cite[Figure
  1]{keurentjes2011} and \cite[Figure 1.7]{voit2018}, the regulatory network
  governing these biotic responses
  
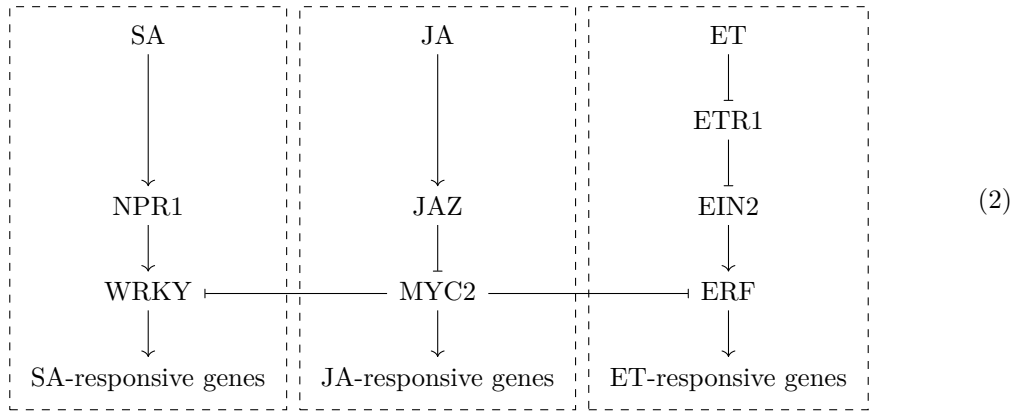
\begin{equation} \label{eq:signaling-reg-nets}
    \begin{tikzcd}[
      column sep=small,
      /tikz/execute at end picture={
        \node (SAres) [rectangle, draw, dashed, fit=(SA) (NPR1) (WRKY) (SAgen)] {};
        \node (JAres) [rectangle, draw, dashed, fit=(JA) (JAZ) (MYC2) (JAgen)] {};
        \node (ETres) [rectangle, draw, dashed, fit=(ET) (ETR1) (EIN2) (ERF) (ETgen)] {};
      }
      ]
      |[alias=SA]| {\mathrm{SA}} & |[alias=JA]| {\mathrm{JA}} & |[alias=ET]| {\mathrm{ET}}\\
      {} & {} & |[alias=ETR1]| {\mathrm{ETR1}}\\
      |[alias=NPR1]| {\mathrm{NPR1}} & |[alias=JAZ]| {\mathrm{JAZ}} & |[alias=EIN2]| {\mathrm{EIN2}}\\
      |[alias=WRKY]| {\mathrm{WRKY}} & |[alias=MYC2]| {\mathrm{MYC2}} & |[alias=ERF]| {\mathrm{ERF}}\\
      |[alias=SAgen]| \text{SA-responsive genes} & |[alias=JAgen]| \text{JA-responsive genes} & |[alias=ETgen]| \text{ET-responsive genes}
      \arrow[from=1-1, to=3-1]
      \arrow[from=3-1, to=4-1]
      \arrow[from=4-1, to=5-1]
      \arrow[from=1-2, to=3-2]
      \arrow[maps to, no head, from=4-2, to=3-2]
      \arrow[from=4-2, to=5-2]
      \arrow[maps to, no head, from=4-1, to=4-2]
      \arrow[maps to, no head, from=4-3, to=4-2]
      \arrow[maps to, no head, from=2-3, to=1-3]
      \arrow[maps to, no head, from=3-3, to=2-3]
      \arrow[from=3-3, to=4-3]
      \arrow[from=4-3, to=5-3]
    \end{tikzcd}
  \end{equation}
  has five discernible subsystems: three signaling processes involving SA, JA,
  and ET, indicated by the dashed boxes, and two interactions \emph{between}
  these processes, namely the inhibitions of WRKY and ERF by MYC2 that
  constitute SA-JA and JA-ET interactions. These subsystems and their functions
  were identified empirically, and it is known that they interact in the
  prescribed manner.

  Following the decomposition identified by the biologists, we build up the
  overall system from smaller subsystems. Namely, we identify five subsystems
  and compose them from left to right. We could equally well have identified
  other subsystems or done the compositions in a different order, yielding an
  equivalent system. Horizontal associativity in the double category
  $\Open{\SgnGraph}$, constructed in \cref{prop:open-sgn-graphs}, ensures that
  any order of composition yields the same result, up to isomorphism.

  We choose to separate out the five subsystems into small, individual
  regulatory networks:
  
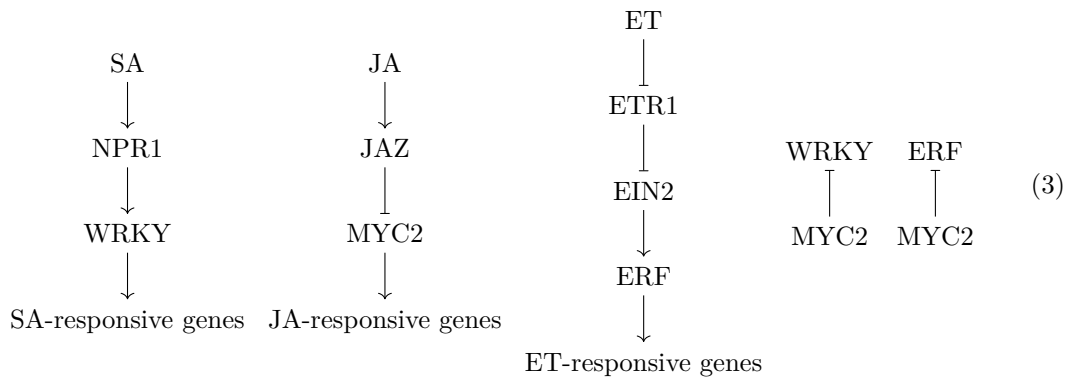
\begin{equation} \label{eq:breakout-signaling}
    \begin{tikzcd}[column sep=small]
      {\mathrm{SA}}\\
      {\mathrm{NPR1}}\\
      {\mathrm{WRKY}}\\
      \text{SA-responsive genes}
      \arrow[from=1-1, to=2-1]
      \arrow[from=2-1, to=3-1]
      \arrow[from=3-1, to=4-1]
    \end{tikzcd}
    \hfill
    \begin{tikzcd}[column sep=small]
      {\mathrm{JA}}\\
      {\mathrm{JAZ}}\\
      {\mathrm{MYC2}}\\
      \text{JA-responsive genes}
      \arrow[from=1-1, to=2-1]
      \arrow[maps to, no head, from=3-1, to=2-1]
      \arrow[from=3-1, to=4-1]
    \end{tikzcd}
    \hfill
    \begin{tikzcd}[column sep=small]
      {\mathrm{ET}}\\
      {\mathrm{ETR1}}\\
      {\mathrm{EIN2}}\\
      {\mathrm{ERF}}\\
      \text{ET-responsive genes}
      \arrow[maps to, no head, from=2-1, to=1-1]
      \arrow[maps to, no head, from=3-1, to=2-1]
      \arrow[from=3-1, to=4-1]
      \arrow[from=4-1, to=5-1]
    \end{tikzcd}
    \hfill
    \begin{tikzcd}[column sep=small]
      {\mathrm{WRKY}}\\
      {\mathrm{MYC2}}
      \arrow[maps to, no head, from=1-1, to=2-1]
    \end{tikzcd}
    \hfill
    \begin{tikzcd}[column sep=small]
      {\mathrm{ERF}}\\
      {\mathrm{MYC2}}
      \arrow[maps to, no head, from=1-1, to=2-1]
    \end{tikzcd}
  \end{equation}
  To turn these into \emph{open} regulatory networks, as defined in
  \cref{prop:open-sgn-cats} above, we will regard the SA signaling subsystem as
  having no inputs and one output as follows.
  
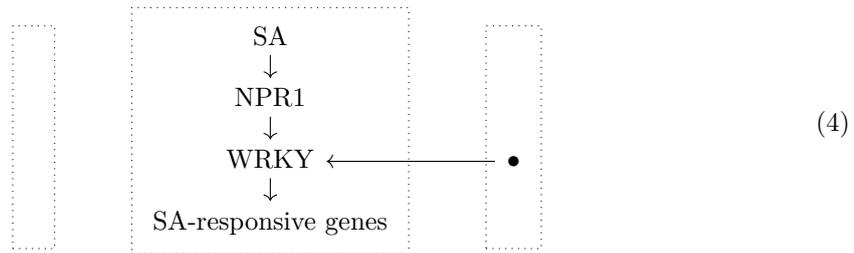
\begin{equation} \label{eq:open-SA-sig}
    \begin{tikzcd}[
      column sep=large, row sep=small,
      /tikz/execute at end picture={
        \node (IN) [rectangle, draw, dotted, fit=(in1) (in2) (in3) (in4)] {};
        \node (SAres) [rectangle, draw, dotted, fit=(SA) (NPR1) (WRKY) (SAgen)] {};
        \node (OUT) [rectangle, draw, dotted, fit=(out1) (out2) (out3) (out4)] {};
      }
      ]
      |[alias=in1]| {} & |[alias=SA]| {\mathrm{SA}} & |[alias=out1]| {}\\
      |[alias=in2]| {} & |[alias=NPR1]| {\mathrm{NPR1}} & |[alias=out2]| {}\\
    	|[alias=in3]| {} & |[alias=WRKY]| {\mathrm{WRKY}} & 	|[alias=out3]| {\bullet}\\
      |[alias=in4]| {} &  |[alias=SAgen]| \text{SA-responsive genes} & |[alias=out4]| {}
      \arrow[from=1-2, to=2-2]
      \arrow[from=2-2, to=3-2]
      \arrow[from=3-2, to=4-2]
      \arrow[from=3-3, to=3-2]
    \end{tikzcd}
  \end{equation}
  Similarly, ET signaling will admit one input and no outputs.
  
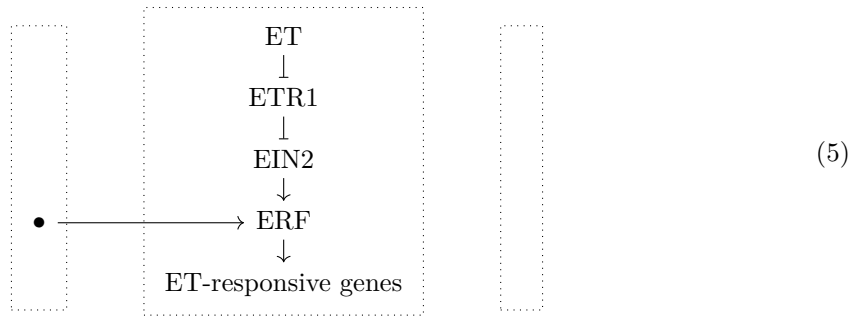
\begin{equation} \label{eq:open-ET-sig}
    \begin{tikzcd}[
      column sep=large, row sep=small,
      /tikz/execute at end picture={
        \node (IN) [rectangle, draw, dotted, fit=(in1) (in2) (in3) (in4) (in5)] {};
        \node (ETres) [rectangle, draw, dotted, fit=(ET) (ETR1) (EIN2) (ERF) (ETgen)] {};
        \node (OUT) [rectangle, draw, dotted, fit=(out1) (out2) (out3) (out4) (out5)] {};
      }
      ]
    	|[alias=in1]| {} & |[alias=ET]| {\mathrm{ET}} & |[alias=out1]| {}\\
    	|[alias=in2]| {} & |[alias=ETR1]| {\mathrm{ETR1}} & |[alias=out2]| {}\\
    	|[alias=in3]| {} & |[alias=EIN2]| {\mathrm{EIN2}} & |[alias=out3]| {}\\
    	|[alias=in4]| {\bullet} & |[alias=ERF]| {\mathrm{ERF}} & 	|[alias=out4]| {}\\
    	|[alias=in5]| {} &  |[alias=ETgen]| \text{ET-responsive genes} & |[alias=out5]| {}
      \arrow[maps to, no head, from=2-2, to=1-2]
      \arrow[maps to, no head, from=3-2, to=2-2]
      \arrow[from=3-2, to=4-2]
      \arrow[from=4-2, to=5-2]
      \arrow[from=4-1, to=4-2]
    \end{tikzcd}
  \end{equation}
  JA signaling will admit one input and one output.
  
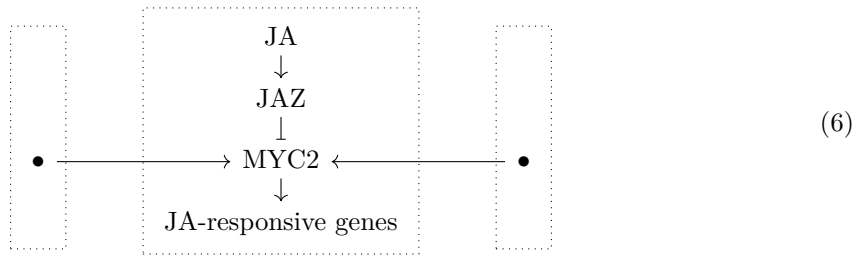
\begin{equation} \label{eq:open-JA-sig}
    \begin{tikzcd}[
      column sep=large, row sep=small,
      /tikz/execute at end picture={
        \node (IN) [rectangle, draw, dotted, fit=(in1) (in2) (in3) (in4)] {};
        \node (SAres) [rectangle, draw, dotted, fit=(JA) (JAZ) (MYC2) (JAgen)] {};
        \node (OUT) [rectangle, draw, dotted, fit=(out1) (out2) (out3) (out4)] {};
      }
      ]
      |[alias=in1]| {} & |[alias=JA]| {\mathrm{JA}} & |[alias=out1]| {}\\
      |[alias=in2]| {} & |[alias=JAZ]| {\mathrm{JAZ}} & |[alias=out2]| {}\\
    	|[alias=in3]| {\bullet} & |[alias=MYC2]| {\mathrm{MYC2}} & |[alias=out3]| {\bullet}\\
      |[alias=in4]| {} &  |[alias=JAgen]| \text{JA-responsive genes} & |[alias=out4]| {}
      \arrow[from=1-2, to=2-2]
      \arrow[maps to, no head, from=3-2, to=2-2]
      \arrow[from=3-2, to=4-2]
      \arrow[from=3-1, to=3-2]
      \arrow[from=3-3, to=3-2]
    \end{tikzcd}
  \end{equation}
  Finally, the inhibition of WRKY by MYC2 will have one input to WRKY and one
  output from MYC2, whereas the inhibition of ERF by MYC2 will have one output
  from ERF and one input from MYC2.
  
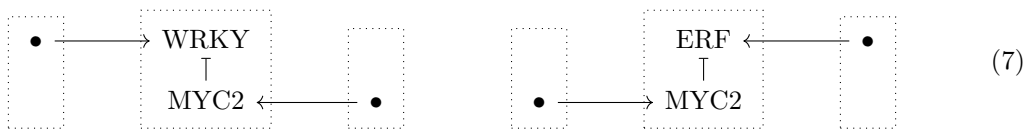
\begin{equation} \label{eq:open-MYC2-inhib}
    \begin{tikzcd}[
      column sep=large, row sep=small,
      /tikz/execute at end picture={
        \node (IN) [rectangle, draw, dotted, fit=(in1) (in2)] {};
        \node (inhib) [rectangle, draw, dotted, fit=(WRKY) (MYC2)] {};
        \node (OUT) [rectangle, draw, dotted, fit=(out1) (out2)] {};
      }
      ]
      |[alias=in1]| {\bullet} & |[alias=WRKY]| {\mathrm{WRKY}} & |[alias=out1]| {}\\
      |[alias=in2]| {} & |[alias=MYC2]| {\mathrm{MYC2}} & |[alias=out2]| {\bullet}
      \arrow[maps to, no head, from=1-2, to=2-2]
      \arrow[from=1-1, to=1-2]
      \arrow[from=2-3, to=2-2]
    \end{tikzcd}
    \qquad\qquad
    \begin{tikzcd}[
      column sep=large, row sep=small,
      /tikz/execute at end picture={
        \node (IN) [rectangle, draw, dotted, fit=(in1) (in2)] {};
        \node (inhib) [rectangle, draw, dotted, fit=(ERF) (MYC2)] {};
        \node (OUT) [rectangle, draw, dotted, fit=(out1) (out2)] {};
      }
      ]
      |[alias=in1]| {} & |[alias=ERF]| {\mathrm{ERF}} & |[alias=out1]| {\bullet}\\
      |[alias=in2]| {\bullet} & |[alias=MYC2]| {\mathrm{MYC2}} & |[alias=out2]| {}
      \arrow[maps to, no head, from=1-2, to=2-2]
      \arrow[from=1-3, to=1-2]
      \arrow[from=2-1, to=2-2]
    \end{tikzcd}
  \end{equation}
  The only thing left to do is compose (from left to right) these open
  regulatory nets. Starting with \cref{eq:open-JA-sig}, we compose it with the
  left part of \cref{eq:open-MYC2-inhib} by first putting them side by side and
  identifying the intermediate sets
  
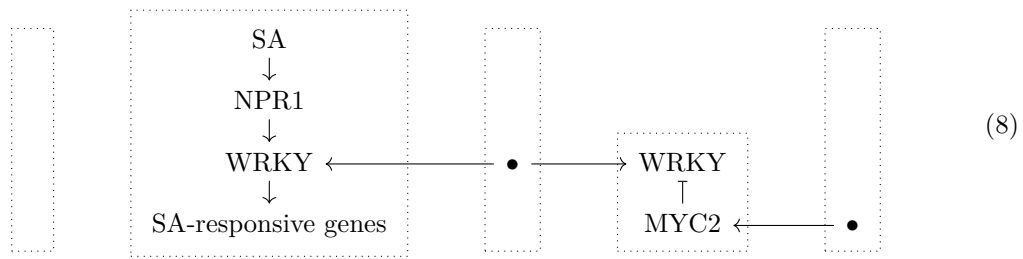
\begin{equation} \label{eq:open-SA-sig-composed-WRKY-inhib}
    \begin{tikzcd}[
      column sep=large, row sep=small,
      /tikz/execute at end picture={
        \node (IN) [rectangle, draw, dotted, fit=(in1) (in2) (in3) (in4)] {};
        \node (SAres) [rectangle, draw, dotted, fit=(SA) (NPR1) (WRKY1) (SAgen)] {};
        \node (mid) [rectangle, draw, dotted, fit=(mid1) (mid2) (mid3) (mid4)] {};
        \node (inhib) [rectangle, draw, dotted, fit=(WRKY2) (MYC2)] {};
        \node (OUT) [rectangle, draw, dotted, fit=(out1) (out2) (out3) (out4)] {};
      }
      ]
      |[alias=in1]| {} & |[alias=SA]| {\mathrm{SA}} & |[alias=mid1]| {} & {} & |[alias=out1]| {}\\
      |[alias=in2]| {} & |[alias=NPR1]| {\mathrm{NPR1}} & |[alias=mid2]| {} & {} & |[alias=out2]| {}\\
    	|[alias=in3]| {} & |[alias=WRKY1]| {\mathrm{WRKY}} & 	|[alias=mid3]| {\bullet} & |[alias=WRKY2]| {\mathrm{WRKY}} & |[alias=out3]| {}\\
      |[alias=in4]| {} &  |[alias=SAgen]| \text{SA-responsive genes} & |[alias=mid4]| {} & |[alias=MYC2]| {\mathrm{MYC2}} & |[alias=out4]| {\bullet}
      \arrow[from=1-2, to=2-2]
      \arrow[from=2-2, to=3-2]
      \arrow[from=3-2, to=4-2]
      \arrow[from=3-3, to=3-2]
      \arrow[maps to, no head, from=3-4, to=4-4]
      \arrow[from=3-3, to=3-4]
      \arrow[from=4-5, to=4-4]
    \end{tikzcd}
  \end{equation}
  which we then remove, together with its outgoing arrows, while identifying the
  vertexes they connect.
  
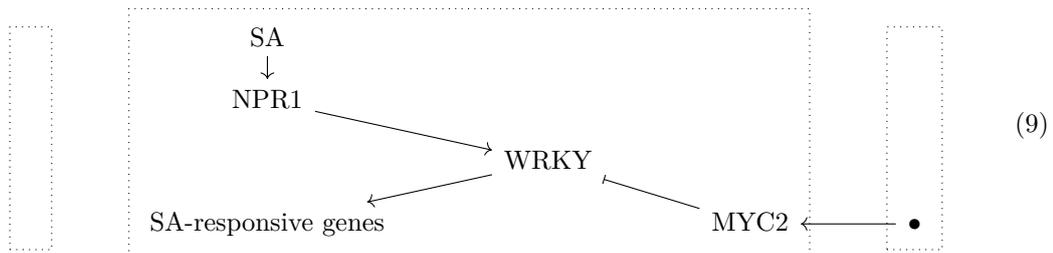
\begin{equation} \label{eq:open-SA-sig-intermediate-WRKY-inhib}
    \begin{tikzcd}[
      column sep=large, row sep=small,
      /tikz/execute at end picture={
        \node (IN) [rectangle, draw, dotted, fit=(in1) (in2) (in3) (in4)] {};
        \node (all) [rectangle, draw, dotted, fit=(SA) (SAgen) (mid3) (MYC2)] {};
        \node (OUT) [rectangle, draw, dotted, fit=(out1) (out2) (out3) (out4)] {};
      }
      ]
      |[alias=in1]| {} & |[alias=SA]| {\mathrm{SA}} & |[alias=mid1]| {} & {} & |[alias=out1]| {}\\
      |[alias=in2]| {} & |[alias=NPR1]| {\mathrm{NPR1}} & |[alias=mid2]| {} & {} & |[alias=out2]| {}\\
    	|[alias=in3]| {} & |[alias=WRKY1]| {} & 	|[alias=mid3]| {\mathrm{WRKY}} & |[alias=WRKY2]| {} & |[alias=out3]| {}\\
      |[alias=in4]| {} &  |[alias=SAgen]| \text{SA-responsive genes} & |[alias=mid4]| {} & |[alias=MYC2]| {\mathrm{MYC2}} & |[alias=out4]| {\bullet}
      \arrow[from=1-2, to=2-2]
      \arrow[from=2-2, to=3-3]
      \arrow[from=3-3, to=4-2]
      \arrow[maps to, no head, from=3-3, to=4-4]
      \arrow[from=4-5, to=4-4]
    \end{tikzcd}
  \end{equation}
  Rearranging the picture slightly for a better visualization, we obtain the
  open regulatory network that results from the composition:
  \begin{equation} \label{eq:open-SA-sig-final-WRKY-inhib}
    \begin{tikzcd}[
      column sep=large, row sep=small,
      /tikz/execute at end picture={
        \node (IN) [rectangle, draw, dotted, fit=(in1) (in2) (in3) (in4)] {};
        \node (all) [rectangle, draw, dotted, fit=(SA) (SAgen) (MYC2)] {};
        \node (OUT) [rectangle, draw, dotted, fit=(out1) (out2) (out3) (out4)] {};
      }
      ]
      |[alias=in1]| {} & |[alias=SA]| {\mathrm{SA}} & {} & |[alias=out1]| {}\\
      |[alias=in2]| {} & |[alias=NPR1]| {\mathrm{NPR1}} & {} & |[alias=out2]| {}\\
    	|[alias=in3]| {} & |[alias=WRKY]| {\mathrm{WRKY}} & |[alias=MYC2]| {\mathrm{MYC2}} & |[alias=out3]| {\bullet}\\
      |[alias=in4]| {} &  |[alias=SAgen]| \text{SA-responsive genes} & {} & |[alias=out4]| {}
      \arrow[from=1-2, to=2-2]
      \arrow[from=2-2, to=3-2]
      \arrow[from=3-2, to=4-2]
      \arrow[maps to, no head, from=3-2, to=3-3]
      \arrow[from=3-4, to=3-3]
    \end{tikzcd}.
  \end{equation}
  Next, composing \cref{eq:open-SA-sig-final-WRKY-inhib,eq:open-JA-sig} yields
  \begin{equation} \label{eq:open-SA-JA}
    \begin{tikzcd}[
      column sep=large, row sep=small,
      /tikz/execute at end picture={
        \node (IN) [rectangle, draw, dotted, fit=(in1) (in2) (in3) (in4)] {};
        \node (all) [rectangle, draw, dotted, fit=(SA) (SAgen) (JA) (JAZ) (MYC2) (JAgen)] {};
        \node (OUT) [rectangle, draw, dotted, fit=(out1) (out2) (out3) (out4)] {};
      }
      ]
      |[alias=in1]| {} & |[alias=SA]| {\mathrm{SA}} & |[alias=JA]| {\mathrm{JA}} & |[alias=out1]| {}\\
      |[alias=in2]| {} & |[alias=NPR1]| {\mathrm{NPR1}} & |[alias=JAZ]| {\mathrm{JAZ}} & |[alias=out2]| {}\\
    	|[alias=in3]| {} & |[alias=WRKY]| {\mathrm{WRKY}} & |[alias=MYC2]| {\mathrm{MYC2}} & |[alias=out3]| {\bullet}\\
      |[alias=in4]| {} &  |[alias=SAgen]| \text{SA-responsive genes} & |[alias=JAgen]| \text{JA-responsive genes} & |[alias=out4]| {}
      \arrow[from=1-2, to=2-2]
      \arrow[from=2-2, to=3-2]
      \arrow[from=3-2, to=4-2]
      \arrow[from=1-3, to=2-3]
      \arrow[maps to, no head, from=3-3, to=2-3]
      \arrow[from=3-3, to=4-3]
      \arrow[maps to, no head, from=3-2, to=3-3]
      \arrow[from=3-4, to=3-3]
    \end{tikzcd}.
  \end{equation}
  Then composing \cref{eq:open-SA-JA} with the right part of
  \cref{eq:open-MYC2-inhib} yields
  \begin{equation} \label{eq:open-SA-JA-ERF}
    \begin{tikzcd}[
      column sep=small, row sep=small,
      /tikz/execute at end picture={
        \node (IN) [rectangle, draw, dotted, fit=(in1) (in2) (in3) (in4)] {};
        \node (all) [rectangle, draw, dotted, fit=(SA) (SAgen) (JA) (JAZ) (MYC2) (JAgen) (ERF)] {};
        \node (OUT) [rectangle, draw, dotted, fit=(out1) (out2) (out3) (out4)] {};
      }
      ]
      |[alias=in1]| {} & |[alias=SA]| {\mathrm{SA}} & |[alias=JA]| {\mathrm{JA}} & {} & |[alias=out1]| {}\\
      |[alias=in2]| {} & |[alias=NPR1]| {\mathrm{NPR1}} & |[alias=JAZ]| {\mathrm{JAZ}} & {} & |[alias=out2]| {}\\
    	|[alias=in3]| {} & |[alias=WRKY]| {\mathrm{WRKY}} & |[alias=MYC2]| {\mathrm{MYC2}} & |[alias=ERF]| {\mathrm{ERF}} & |[alias=out3]| {\bullet}\\
      |[alias=in4]| {} &  |[alias=SAgen]| \text{SA-responsive genes} & |[alias=JAgen]| \text{JA-responsive genes} & {} & |[alias=out4]| {}
      \arrow[from=1-2, to=2-2]
      \arrow[from=2-2, to=3-2]
      \arrow[from=3-2, to=4-2]
      \arrow[from=1-3, to=2-3]
      \arrow[maps to, no head, from=3-3, to=2-3]
      \arrow[from=3-3, to=4-3]
      \arrow[maps to, no head, from=3-2, to=3-3]
      \arrow[maps to, no head, from=3-4, to=3-3]
      \arrow[from=3-5, to=3-4]
    \end{tikzcd}.
  \end{equation}
  Finally, composing \cref{eq:open-SA-JA-ERF,eq:open-ET-sig} yields
  \begin{equation} \label{eq:open-SA-JA-ERF}
    \begin{tikzcd}[
      column sep=small, row sep=small,
      /tikz/execute at end picture={
        \node (IN) [rectangle, draw, dotted, fit=(in1) (in2) (in3) (in4) (in5)] {};
        \node (all) [rectangle, draw, dotted, fit=(SA) (SAgen) (JA) (JAZ) (MYC2) (JAgen) (ET) (ETR1) (EIN2) (ERF) (ETgen)] {};
        \node (OUT) [rectangle, draw, dotted, fit=(out1) (out2) (out3) (out4) (out5)] {};
      }
      ]
      |[alias=in1]| {} & |[alias=SA]| {\mathrm{SA}} & |[alias=JA]| {\mathrm{JA}} & |[alias=ET]| {\mathrm{ET}} & |[alias=out1]| {}\\
      |[alias=in2]| {} & {} & {} & |[alias=ETR1]| {\mathrm{ETR1}} & |[alias=out2]| {}\\
      |[alias=in3]| {} & |[alias=NPR1]| {\mathrm{NPR1}} & |[alias=JAZ]| {\mathrm{JAZ}} & |[alias=EIN2]| {\mathrm{EIN2}} & |[alias=out3]| {}\\
    	|[alias=in4]| {} & |[alias=WRKY]| {\mathrm{WRKY}} & |[alias=MYC2]| {\mathrm{MYC2}} & |[alias=ERF]| {\mathrm{ERF}} & |[alias=out4]| {}\\
      |[alias=in5]| {} &  |[alias=SAgen]| \text{SA-responsive genes} & |[alias=JAgen]| \text{JA-responsive genes} & |[alias=ETgen]| \text{ET-responsive genes} & |[alias=out5]| {}
      \arrow[from=1-2, to=3-2]
      \arrow[from=3-2, to=4-2]
      \arrow[from=4-2, to=5-2]
      \arrow[from=1-3, to=3-3]
      \arrow[maps to, no head, from=4-3, to=3-3]
      \arrow[from=4-3, to=5-3]
      \arrow[maps to, no head, from=4-2, to=4-3]
      \arrow[maps to, no head, from=4-4, to=4-3]
      \arrow[maps to, no head, from=2-4, to=1-4]
      \arrow[maps to, no head, from=3-4, to=2-4]
      \arrow[from=3-4, to=4-4]
      \arrow[from=4-4, to=5-4]
    \end{tikzcd}
  \end{equation}
  which corresponds to the open regulatory network governing the biotic
  responses (cf.\ \cref{eq:signaling-reg-nets}). As mentioned at the beginning
  of the example, although we made a choice of subsystems and a choice in the
  order of compositions, any other choice will give an isomorphic result by the
  horizontal associativity of the double category $\Open{\SgnGraph}$.
\end{example}

\subsection{Refining regulatory networks using signed categories and functors}
\label{sec:sgn-cat}

While morphisms of signed graph have their uses, they do not capture the
important idea of \emph{refining} regulatory networks, in which an interaction
in one network is realized as a composite of several interactions in another. To
express refinement, we must generalize our notion of morphism from graph
homomorphisms to functors. This, in turn, requires the concept of a \emph{signed
  category}.

\begin{definition}[Signed categories] \label{def:sgn-category}
  The \define{category of signed categories} is the slice category
  \begin{equation*}
    \SgnCat \coloneqq \Cat/\Sgn,
  \end{equation*}
  where $\Cat$ is the category of small categories and the group of signs,
  $\Sgn$, is regarded as a category with one object and two morphisms.
\end{definition}

Unpacking the definition, a \define{signed category} is a category $\cat{C}$ in
which every morphism $f$ is assigned a sign $\sgn(f) \in \{1,-1\}$ in a
functorial way, meaning that
\begin{equation*}
  \sgn(x_0 \xrightarrow{f_1} x_1 \xrightarrow{f_2} \cdots \xrightarrow{f_n} x_n)
    = \prod_{i=1}^n \sgn(f_i)
\end{equation*}
for every $n \geq 0$ and every sequence of composable morphisms
$f_1, \dots, f_n$. In particular ($n=0$), the identity morphisms have positive
sign. A \define{morphism of signed categories}, or \define{signed functor}, is a
functor $F: \cat{C} \to \cat{D}$ between signed categories that preserves the
signs, meaning that
\begin{equation*}
  \sgn_{\cat{D}}(F(f)) = \sgn_{\cat{C}}(f)
\end{equation*}
for every morphism $f$ in $\cat{C}$.

Since our aim is to have a more flexible notion of morphism between signed
graphs, we will mostly restrict ourselves to those signed categories that are
freely generated by a signed graph. The \define{free signed category} or
\define{signed path category} functor
\begin{equation*}
  \Path: \SgnGraph \to \SgnCat
\end{equation*}
sends a signed graph $X$ to the signed category $\Path(X)$ having
\begin{itemize}[noitemsep]
  \item as objects, the vertices of $X$;
  \item as morphisms from $x$ to $y$, the paths in $X$ from $x$ to $y$, whose
    sign is defined to be the product of the signs of the edges comprising the
    path.
\end{itemize}
Composition of paths is by concatenation, which clearly preserves the sign. The
identity morphism at $x$ is the empty path at $x$, which has positive sign. The
functor $\Path$ on signed graphs is completely analogous to the usual free
category functor on graphs, and as such is a left adjoint to the forgetful
functor from signed categories to signed graphs. Signed categories are likewise
algebras for the corresponding monad on the category $\SgnGraph$.

By convention, if $X$ and $Y$ are signed graphs, we say that a \define{signed
functor} from $X$ to $Y$ is a signed functor $F: \Path(X) \to \Path(Y)$ between
the corresponding signed path categories. Since the morphisms of $\Path(X)$ are
freely generated by the edges in $X$, a signed functor from $X$ to $Y$ is
uniquely determined by a morphism of signed graphs from $X$ to the underlying
signed graph of $\Path(Y)$. This means that each edge in $X$ is sent to an
appropriately signed \emph{path} of edges in $Y$, which can be regarded as a
refinement of the relationship that the edge represents. 

\begin{definition}[Category of refinements]\label{def:refine-cat}
  The category $\SgnGraph_{\Path}$ has as objects signed graphs and as morphisms
  signed functors between them; in other words, it is the Kleisli category for
  the free signed category monad on $\SgnGraph$.
\end{definition}

We now have a precise language with which to classify network motifs and their
occurrences. As a first example, Alon identifies four types of \emph{incoherent
  feedforward loop (FFL)} involving three components,
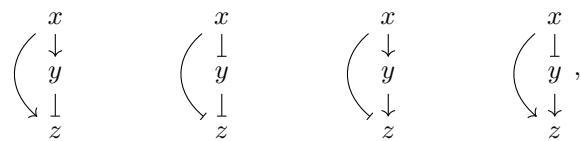
\begin{equation*}
  \begin{tikzcd}[row sep=small]
    x && x && x && x \\
    y && y && y && y \\
    z && z && z && z
    \arrow[from=1-1, to=2-1]
    \arrow[maps to, no head, from=3-1, to=2-1]
    \arrow[curve={height=18pt}, from=1-1, to=3-1]
    \arrow[maps to, no head, from=2-3, to=1-3]
    \arrow[curve={height=-18pt}, maps to, no head, from=3-3, to=1-3]
    \arrow[maps to, no head, from=3-3, to=2-3]
    \arrow[from=1-5, to=2-5]
    \arrow[from=2-5, to=3-5]
    \arrow[curve={height=-18pt}, maps to, no head, from=3-5, to=1-5]
    \arrow[curve={height=18pt}, from=1-7, to=3-7]
    \arrow[maps to, no head, from=2-7, to=1-7]
    \arrow[from=2-7, to=3-7]
  \end{tikzcd},
\end{equation*}
those of \emph{type 1, 2, 3, and 4}, respectively \cite[Figure 2a]{alon2007}.
Besides having three components, what these motifs have in common is that there
exists a signed functor into each of them from the signed graph
$I_{\pm} \coloneqq \left\{
  \begin{tikzcd}[cramped, sep=small]
  	\bullet & \bullet
  	\arrow[curve={height=-6pt}, from=1-1, to=1-2]
  	\arrow[curve={height=-6pt}, maps to, no head, from=1-2, to=1-1]
  \end{tikzcd}
\right\}$ having two parallel arrows of opposite sign. The network $I_{\pm}$ is
thus the ``generic'' incoherent feedforward loop, in the sense that signed
functors out of it refine the pattern in specific ways. A similar situation
holds for other common network motifs (\cref{tab:motifs}), which motivates the
following definition.

\begin{definition}[Motif instance]
  Given a signed graph $A$, regarded as a motif, an \define{instance} or
  \define{occurrence} of the motif $A$ in a network $X$ is a monic signed
  functor $A \monicto X$.
\end{definition}

\begin{table}
  \centering
  \begin{tabular}{lc}
    \toprule
    Motif & Generic instance \\
    \midrule \addlinespace[0.5em]
    Positive autoregulation &
      $L_+ \coloneqq \left\{
        \begin{tikzcd}[cramped]
          \bullet \arrow[loop, out=30, in=330, looseness=5]
        \end{tikzcd}
      \right\}$ \\ \addlinespace[0.5em]
    Negative autoregulation &
      $L_- \coloneqq \left\{
        \begin{tikzcd}[cramped]
          \bullet \arrow[loop, in=30, out=330, maps to, no head, looseness=5]
        \end{tikzcd}
      \right\}$ \\ \addlinespace[0.5em]
    Coherent feedforward loop &
      $I_{++} \coloneqq \left\{
        \begin{tikzcd}[cramped, sep=small]
          \bullet & \bullet
          \arrow[curve={height=-6pt}, from=1-1, to=1-2]
          \arrow[curve={height=6pt}, from=1-1, to=1-2]
        \end{tikzcd}
      \right\}$ \\ \addlinespace[0.5em]
    Incoherent feedforward loop &
      $I_{\pm} \coloneqq \left\{
        \begin{tikzcd}[cramped, sep=small]
          \bullet & \bullet
          \arrow[curve={height=-6pt}, from=1-1, to=1-2]
          \arrow[curve={height=-6pt}, maps to, no head, from=1-2, to=1-1]
        \end{tikzcd}
      \right\}$ \\ \addlinespace[0.5em]
    Positive feedback loop &
      $L_{++} \coloneqq \left\{
        \begin{tikzcd}[cramped, sep=small]
          \bullet & \bullet
          \arrow[curve={height=-6pt}, from=1-1, to=1-2]
          \arrow[curve={height=-6pt}, from=1-2, to=1-1]
        \end{tikzcd}
      \right\}$ \\ \addlinespace[0.5em]
    Negative feedback loop &
      $L_{\pm} \coloneqq \left\{
        \begin{tikzcd}[cramped, sep=small]
          \bullet & \bullet
          \arrow[curve={height=-6pt}, from=1-1, to=1-2]
          \arrow[curve={height=6pt}, maps to, no head, from=1-1, to=1-2]
        \end{tikzcd}
      \right\}$ \\ \addlinespace[0.5em]
    Double-negative feedback loop &
      $L_{--} \coloneqq \left\{
        \begin{tikzcd}[cramped, sep=small]
          \bullet & \bullet
          \arrow[curve={height=6pt}, maps to, no head, from=1-2, to=1-1]
          \arrow[curve={height=6pt}, maps to, no head, from=1-1, to=1-2]
        \end{tikzcd}
      \right\}$ \\ \addlinespace[0.5em]
    \bottomrule
  \end{tabular}
  \caption{Common motifs in biochemical regulation networks
    \cite{alon2007,tyson2010}}
  \label{tab:motifs}
\end{table}

Note that a signed functor is a monomorphism exactly when the functor is an
embedding of categories, i.e., an injective-on-objects, faithful functor.
Requiring the functor in the definition to be monic excludes ``degenerate
instances'' of motifs where vertices or edges are identified.

Now, should the incoherent FFL be regarded as a network motif, or is it the more
specific types, such as the incoherent FFL of type 1, that are motifs? From our
point of view, they are all equally motifs but they have different degrees of
specificity, and the functorial language clarifies how motifs are iteratively
refined. Specifically, an instance of an incoherent FFL of type 1 in a network
$X$ also gives an instance of an incoherent FFL in $X$ (of unspecified type),
simply by composing the monomorphisms involved:
\begin{equation*}
  I_{\pm} \cong
  \left\{
    \begin{tikzcd}[cramped, sep=small]
      x & z
      \arrow[curve={height=-6pt}, from=1-1, to=1-2]
      \arrow[curve={height=-6pt}, maps to, no head, from=1-2, to=1-1]
    \end{tikzcd}
  \right\}
  \quad\monicto\quad
  \left\{
    \begin{tikzcd}[cramped, sep=small]
      x & y & z
      \arrow[curve={height=-12pt}, from=1-1, to=1-3]
      \arrow[from=1-1, to=1-2]
      \arrow[maps to, no head, from=1-3, to=1-2]
    \end{tikzcd}
  \right\}
  \quad\monicto\quad X.
\end{equation*}
Similarly, in the notation of \cref{tab:motifs}, any instance of double-negative
feedback ($L_{--}$) also gives an instance of positive autoregulation ($L_+$)
\cite{crews2009}, via the monomorphism $L_+ \monicto L_{--}$ that sends the
positive loop to the double-negative 2-cycle.

For any choice of motif $A$, the mapping that sends a regulatory network $X$ to
the set of occurrences of $A$ in $X$ is a representable functor
\begin{equation*}
  \Hom(A, -): (\SgnGraph_{\Path})_m \to \Set,
\end{equation*}
where $(\SgnGraph_{\Path})_m$ denotes the wide subcategory of monomorphisms in
$\SgnGraph_{\Path}$. The functorality implies that, for any motif $A$, a
monomorphism between regulatory networks induces a map between instances of $A$
in those networks.

We now extend the construction of open signed graphs to open signed categories.

\begin{proposition}
  The category of signed categories is complete and cocomplete.
\end{proposition}
\begin{proof}
  Because the category $\Cat$ is complete and cocomplete \cite[Proposition
  3.5.6]{riehl2016}, its slice $\SgnCat = \Cat/\Sgn$ is also \cite[Proposition
  3.5.5]{riehl2016}.
\end{proof}

\begin{proposition}[Open signed categories] \label{prop:open-sgn-cats}
  There is a symmetric monoidal double category of open signed categories,
  $\Open{\SgnCat}$, having
  \begin{itemize}[noitemsep]
    \item as objects, sets $A, B, C, \dots$;
    \item as vertical morphisms, functions $f: A \to B$;
    \item as horizontal morphisms, \define{open signed categories}, which consist of
      a signed category $\cat{C}$ together with a cospan of sets
      $A_0 \xrightarrow{\ell_0} \Ob(\cat{C}) \xleftarrow{\ell_1} A_1$;
    \item as cells, \define{morphisms of open signed categories}
      $(\cat{C},\ell_0,\ell_1) \to (\cat{D},m_0,m_1)$, which consist of a signed
      functor $F: \cat{C} \to \cat{D}$ along with functions $f_i: A_i \to B_i$,
      $i=0,1$, making the diagram commute:
      \begin{equation*}
        \begin{tikzcd}
          {A_0} & {\mathrm{Ob}(\mathsf{C})} & {A_1} \\
          {B_0} & {\mathrm{Ob}(\mathsf{D})} & {B_1}
          \arrow["{\ell_0}", from=1-1, to=1-2]
          \arrow["{f_0}"', from=1-1, to=2-1]
          \arrow["{\ell_1}"', from=1-3, to=1-2]
          \arrow["{\mathrm{Ob}(F)}"', from=1-2, to=2-2]
          \arrow["{m_0}"', from=2-1, to=2-2]
          \arrow["{f_1}", from=1-3, to=2-3]
          \arrow["{m_1}", from=2-3, to=2-2]
        \end{tikzcd}.
      \end{equation*}
  \end{itemize}
  Vertical composition is by composition in $\Set$ and in $\SgnCat$. Horizontal
  composition and monoidal products are by pushouts and coproducts in $\SgnCat$,
  respectively, viewing the sets in the feet of cospans as discrete signed
  categories.

  Moreover, the signed path category functor extends to a symmetric monoidal
  double functor
  \begin{equation*}
    \Path: \Open{\SgnGraph} \to \Open{\SgnCat}.
  \end{equation*}
\end{proposition}
\begin{proof}
  We take $\Open{\SgnCat}$ to be the symmetric monoidal double category of
  $L'$-structured cospans for the functor $L' \coloneqq \Disc: \Set \to \SgnCat$
  involved the composite adjunction
  \begin{equation*}
    \begin{tikzcd}
      {\mathsf{Set}} & {\mathsf{SgnCat}} & {=} & {\mathsf{Set}} & {\mathsf{SgnGraph}} & {\mathsf{SgnCat}}
      \arrow[""{name=0, anchor=center, inner sep=0}, "{\mathrm{Disc}}", curve={height=-12pt}, from=1-1, to=1-2]
      \arrow[""{name=1, anchor=center, inner sep=0}, "{\mathrm{Ob}}", curve={height=-12pt}, from=1-2, to=1-1]
      \arrow[""{name=2, anchor=center, inner sep=0}, "{\mathrm{Disc}}", curve={height=-12pt}, from=1-4, to=1-5]
      \arrow[""{name=3, anchor=center, inner sep=0}, "{\mathrm{ev}_V}", curve={height=-12pt}, from=1-5, to=1-4]
      \arrow[""{name=4, anchor=center, inner sep=0}, "U", curve={height=-12pt}, from=1-6, to=1-5]
      \arrow[""{name=5, anchor=center, inner sep=0}, "{\mathrm{Path}}", curve={height=-12pt}, from=1-5, to=1-6]
      \arrow["\dashv"{anchor=center, rotate=-90}, draw=none, from=0, to=1]
      \arrow["\dashv"{anchor=center, rotate=-90}, draw=none, from=2, to=3]
      \arrow["\dashv"{anchor=center, rotate=-90}, draw=none, from=5, to=4]
    \end{tikzcd}.
  \end{equation*}
  On the right hand side, the first adjunction was already used in the proof of
  \cref{prop:open-sgn-graphs}, and the second adjunction is the free-forgetful
  adjunction between signed graphs and signed categories.

  To prove the last statement, we notice that all functors involved in the
  commutative square
  \begin{equation*}
    \begin{tikzcd}
      \Set && \SgnGraph \\
      \Set && \SgnCat
      \arrow[Rightarrow, no head, from=1-1, to=2-1]
      \arrow["{L = \Disc}", from=1-1, to=1-3]
      \arrow["{L' = \Disc}"', from=2-1, to=2-3]
      \arrow["\Path", from=1-3, to=2-3]
    \end{tikzcd}
  \end{equation*}
  are left adjoints, hence preserve colimits. We can therefore appeal to
  \cite[Theorem 4.3]{baez2020} to obtain a symmetric monoidal double functor
  \begin{equation*}
    \Open{\SgnGraph} \cong {}_{L}\Csp(\SgnGraph) \to
      {}_{L'}\Csp(\SgnCat) \cong \Open{\SgnCat}. \qedhere
  \end{equation*}
\end{proof}

\subsection{Mechanistic models as Petri nets with links}
\label{sec:sgn-petri}

However challenging they may be to identify through experiments and data
analysis, regulatory networks still only summarize how the components of a
complex biochemical system interact. Regulatory networks typically include only
a subset of the system's components, and they do not model individual reactions
and processes, only pairwise promoting or inhibiting interactions between
components. In this sense, regulatory networks are not fully mechanistic models,
even if they have a stronger causal interpretation than, say, a correlation
matrix.

By contrast, mechanistic models in biochemistry model individual reactions,
which requires a different formalism. Pictures like the following, adapted from
Voit's review \cite[Figure 4]{voit2013}, are common in systems biology.
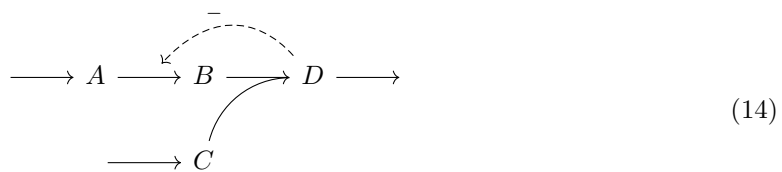
\begin{equation} \label{eq:voit-ex}
  \begin{tikzcd}
    {} & A & B & D & {} \\
    & {} & C
    \arrow[curve={height=-14pt}, no head, from=2-3, to=1-4]
    \arrow[""{name=0, anchor=center, inner sep=0}, from=1-2, to=1-3]
    \arrow[from=1-3, to=1-4]
    \arrow[from=1-1, to=1-2]
    \arrow[from=2-2, to=2-3]
    \arrow[from=1-4, to=1-5]
    \arrow["{-}"', curve={height=24pt}, shorten >=6pt, dashed, from=1-4, to=0]
  \end{tikzcd}
\end{equation}
This diagram possesses two distinctive features. First, directed hyperedges
represent reactions having a number of inputs or outputs different than one.
There are, for example, hyperedges from $B$ \emph{and} $C$ to $D$, from nothing
to $A$ (an inflow), and from $D$ to nothing (an outflow). If, in lieu of
hyperedges, we introduce a second type of vertex, we obtain a structure similar
to a Petri net
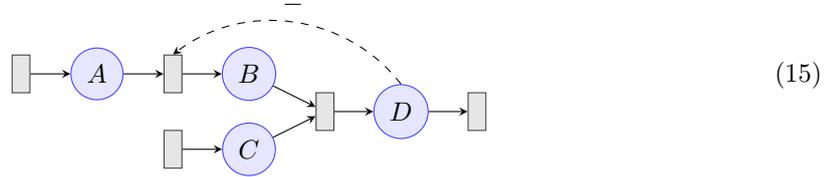
\begin{equation} \label{eq:voit-ex-sgn-petri}
  \begin{tikzpicture}[baseline={(current bounding box.center)}]
    \node[transition] (in1) at (-3,0) {};
    \node[species] (A) at (-2,0) {$A$};
    \node[transition] (AB) at (-1,0) {};
    \node[transition] (in2) at (-1,-1) {};
    \node[species] (B) at (0,0) {$B$};
    \node[species] (C) at (0,-1) {$C$};
    \node[transition] (BCD) at (1,-0.5) {};
    \node[species] (D) at (2,-0.5) {$D$};
    \node[transition] (out1) at (3,-0.5) {};
    \draw[arc] (in1) to (A);
    \draw[arc] (in2) to (C);
    \draw[arc] (A) to (AB);
    \draw[arc] (AB) to (B);
    \draw[arc] (B) to (BCD);
    \draw[arc] (C) to (BCD);
    \draw[arc] (BCD) to (D);
    \draw[arc] (D) to (out1);
    \draw[link, bend right=45] (D.north) to node[midway,above] {$-$} (AB.north);
  \end{tikzpicture}
\end{equation}
but with the second distinctive feature of having \emph{signed links} from the
first type of vertices (\emph{species}) to the second type (\emph{transitions}),
whose signs indicate promotion or inhibition.

In this section, we explain how a Petri net with signed links can provide a
\emph{mechanism} for a regulatory network. This involves constructing a functor
from Petri nets with signed links to signed graphs, approximating the former as
the latter. As a prerequisite, we need a precise definition of a Petri net with
links.

\begin{definition}[Petri net with links]
  The \define{schema for Petri nets with links} is the category
  $\Sch{\PetriLink}$ freely generated by these objects and morphisms:
  
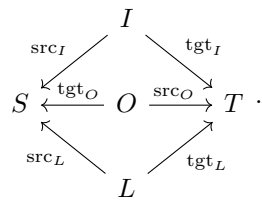
\begin{equation*}
    \begin{tikzcd}
      & I \\
      S & O & T \\
      & L
      \arrow["{\mathrm{src}_L}", from=3-2, to=2-1]
      \arrow["{\mathrm{tgt}_L}"', from=3-2, to=2-3]
      \arrow["{\mathrm{src}_I}"', from=1-2, to=2-1]
      \arrow["{\mathrm{tgt}_I}", from=1-2, to=2-3]
      \arrow["{\mathrm{tgt}_O}"'{pos=0.4}, from=2-2, to=2-1]
      \arrow["{\mathrm{src}_O}"{pos=0.4}, from=2-2, to=2-3]
    \end{tikzcd}.
  \end{equation*}
  A \define{Petri nets with links} is a functor $P: \Sch{\PetriLink} \to \Set$,
  and a \define{morphism} of these is a natural transformation. A morphism
  $\phi: P \to Q$ \define{preserves arities} if the naturality squares
  associated with the morphisms $I \to T$ and $O \to T$ are also pullback
  squares:
  \begin{equation*}
    \begin{tikzcd}
      {P(I)} & {P(T)} \\
      {Q(I)} & {Q(T)}
      \arrow["{P(\mathrm{tgt}_I)}", from=1-1, to=1-2]
      \arrow["{\phi_I}"', from=1-1, to=2-1]
      \arrow["{\phi_T}", from=1-2, to=2-2]
      \arrow["{P(\mathrm{tgt}_I)}"', from=2-1, to=2-2]
      \arrow["\lrcorner"{anchor=center, pos=0.125}, draw=none, from=1-1, to=2-2]
    \end{tikzcd}
    \qquad\qquad
    \begin{tikzcd}
      {P(O)} & {P(T)} \\
      {Q(O)} & {Q(T)}
      \arrow["{P(\mathrm{src}_O)}", from=1-1, to=1-2]
      \arrow["{\phi_O}"', from=1-1, to=2-1]
      \arrow["{\phi_T}", from=1-2, to=2-2]
      \arrow["{P(\mathrm{src}_O)}"', from=2-1, to=2-2]
      \arrow["\lrcorner"{anchor=center, pos=0.125}, draw=none, from=1-1, to=2-2]
    \end{tikzcd}.
  \end{equation*}
  Petri nets with links and their morphisms form the category $\PetriLink$.
\end{definition}

To explicate the definition, a Petri net with links $P$ consists of
\begin{itemize}[noitemsep]
  \item a set $P(S)$ of \define{species};
  \item a set $P(T)$ of \define{transitions};
  \item a set $P(I)$ of \define{input arcs} going from species to transitions,
    via maps \mbox{$P(\src_I): P(I) \to P(S)$} and $P(\tgt_I): P(I) \to P(T)$;
  \item a set $P(O)$ of \define{output arcs} going from transitions to species,
    via maps \mbox{$P(\src_O): P(O) \to P(T)$} and $P(\tgt_O): P(O) \to P(S)$;
    and finally
  \item a set $P(L)$ of \define{links} going from species to transitions, via
    maps \mbox{$P(\src_L): P(L) \to P(S)$} and $P(\tgt_L): P(L) \to P(T)$.
\end{itemize}
The property of preserving arities, called ``etale'' by Kock \cite[\S
3.4]{kock2022}, means that a morphism $\phi: P \to Q$ of Petri nets with links
preserves the input and output arities of all transitions. Namely, for each
transition $t$ in the net $P$ the map $\phi_I: P(I) \to Q(I)$ restricts to a
bijection between the input arcs to $t$ and to $\phi_T(t)$, and similarly the
map $\phi_O: P(O) \to Q(O)$ restricts to a bijection between the output arcs
from $t$ and from $\phi_T(t)$. This property seems appropriate for many
purposes, including in biochemistry, but for mathematical convenience we will
not always assume it.

\begin{remark}[Related literature]
  While not appearing in the literature in precisely this form, our definition
  of a Petri net with links is closely related to several existing concepts.
  Kock has described Petri nets as copresheaves on a category with objects
  $S,T,I,O$ \cite{kock2022}, calling them \emph{whole-grain Petri nets} to
  distinguish them from classical Petri nets, whose semantics are subtly
  different \cite{baez2021}. Meanwhile, the concept of a link is essential to
  \emph{stock and flow diagrams}, originating in the field of system dynamics
  \cite{forrester1961,sterman2000} and recently given a rigorous categorical
  account \cite{baez2022}. The classical literature on Petri nets has also
  considered Petri nets equipped with positive/context and negative/inhibitory
  arcs \cite{agerwala1973}, called \emph{inhibitor nets}
  \cite{hack1976,baldan2004}. Roughly speaking, our Petri nets with signed links
  are related to inhibitor nets (\cite[\mbox{Definition 1}]{baldan2004}) in the
  same way that Kock's whole-grain Petri nets are related to classical Petri
  nets; in both cases, working with copresheaves has technical advantages.
  Finally, we note that the \emph{Petri nets with catalysts} proposed by Baez,
  Foley, and Moeller \cite{baez2019} differ significantly from Petri nets with
  links: the former fix a subset of the species to be catalysts throughout the
  net, whereas the latter uses links to make catalyzation specific to individual
  reactions.
\end{remark}

\begin{remark}[Petri nets as typed graphs]
  Like bare Petri nets, Petri nets with links can be described as graphs with
  two types of vertices. To see this, take the graph
  \begin{equation*}
    T_{\PetriLink} \coloneqq \left\{
      \begin{tikzcd}
        S && T
        \arrow["I", curve={height=-18pt}, from=1-1, to=1-3]
        \arrow["L"', curve={height=18pt}, from=1-1, to=1-3]
        \arrow["O"{description}, from=1-3, to=1-1]
      \end{tikzcd}
    \right\}
  \end{equation*}
  with vertices $S$ and $T$ and edges $I$, $O$, and $L$. The category of Petri
  nets with links and natural transformations is isomorphic to the slice
  category $\Graph/T_{\PetriLink}$. Moreover, the schema category
  $\Sch{\PetriLink}$ is isomorphic to the category of elements of the functor
  $T_{\PetriLink}: \Sch{\Graph} \to \Set$, exemplifying a general fact about
  slices of copresheaf categories \cite[Remark p.\ 303]{street2000}.
\end{remark}

Petri nets with \emph{signed} links are defined analogously to signed graphs
(\cref{def:sgn-graph}).

\begin{definition}[Petri nets with signed links] \label{def:sgn-petri}
  The \define{category of Petri nets with signed links} is the slice category
  \begin{equation*}
    \SgnPetri \coloneqq \PetriLink/P_\Sgn,
  \end{equation*}
  where $P_\Sgn$ is the Petri net with links having one species, one transition,
  one input arc, one output arc, and two links, namely the elements of $\Sgn$.
\end{definition}

Incidentally, the morphism $P \to P_\Sgn$ defining a Petri net with signed links
does \emph{not} preserve arities unless every transition in $P$ has exactly one
input and one output.

We now turn to the main construction of this section, a functor that sends a
Petri net with signed links to the signed graph of interactions implied by the
net. On the example in \cref{eq:voit-ex,eq:voit-ex-sgn-petri}, this functor
gives the signed graph
\begin{equation} \label{eq:voit-ex-sgn-graph}
  \begin{tikzcd}
    A \arrow[loop below, maps to, no head, looseness=7]
    & B \arrow[loop below, maps to, no head, looseness=7]
    & D \arrow[loop below, maps to, no head, looseness=7] \\
    & C \arrow[loop below, maps to, no head, looseness=7]
    \arrow[from=1-1, to=1-2]
    \arrow[from=1-2, to=1-3]
    \arrow[from=2-2, to=1-3]
    \arrow[curve={height=-6pt}, maps to, no head, from=1-2, to=1-3]
    \arrow[curve={height=18pt}, from=1-3, to=1-1]
  \end{tikzcd}.
\end{equation}
In general, the resulting signed graph has the Petri net's species as vertices,
and has signed edges of the following four types:
\begin{enumerate}[(a),noitemsep]
  \item for every input-output pair to a transition, a positive edge from input
    to output;
  \item for every input to a transition, a negative self loop, representing
    consumption by the reaction;
  \item for every signed link, an edge of \emph{opposite} sign going from the
    linked species to each input to the linked transition;
  \item for every signed link, an edge of \emph{equal} sign going from the
    linked species to each output from the linked transition.
\end{enumerate}
All four cases are visible in the example of \cref{eq:voit-ex-sgn-graph}. Each
of these cases is detected by a representable functor of the form
\begin{equation*}
  \Hom_{\PetriLink}(P, -): \PetriLink \to \Set,
\end{equation*}
where $P$ is one of the four Petri nets with links shown in
\cref{fig:petri-queries}. Each net represents a ``pattern'' whose matches
produce edges in the signed graph.

\begin{figure}
  \centering
  \begin{subfigure}[b]{0.24\textwidth}
    \centering
    \begin{tikzpicture}
      \node[species] (in) at (-1,0) {};
      \node[transition] (t) at (0,0) {};
      \node[species] (out) at (1,0) {};
      \draw[arc] (in) to (t);
      \draw[arc] (t) to (out);
    \end{tikzpicture}
    \caption{Input-output pair to transition}
  \end{subfigure}
  \hfill
  \begin{subfigure}[b]{0.24\textwidth}
    \centering
    \begin{tikzpicture}
      \node[species] (in) at (-1,0) {};
      \node[transition] (t) at (0,0) {};
      \draw[arc] (in) to (t);
    \end{tikzpicture}
    \caption{Input to transition}
  \end{subfigure}
  \hfill
  \begin{subfigure}[b]{0.24\textwidth}
    \centering
    \begin{tikzpicture}
      \node[species] (in) at (-1,0) {};
      \node[transition] (t) at (0,0) {};
      \node[species] (s) at (1,0) {};
      \draw[arc] (in) to (t);
      \draw[link, bend right] (s.north) to (t.north);
    \end{tikzpicture}
    \caption{Input to transition with incident link}
  \end{subfigure}
  \hfill
  \begin{subfigure}[b]{0.24\textwidth}
    \centering
    \begin{tikzpicture}
      \node[transition] (t) at (0,0) {};
      \node[species] (out) at (1,0) {};
      \node[species] (s) at (2,0) {};
      \draw[arc] (t) to (out);
      \draw[link, bend right] (s.north) to (t.north);
    \end{tikzpicture}
    \caption{Output from transition with incident link}
  \end{subfigure}
  \caption{Four different Petri nets with links. For each of these instances
    $P$, evaluating the representable functor
    $\Hom_{\PetriLink}(P,-): \PetriLink \to \Set$ gives the edges for one case
    in the case analysis that defines the functor from Petri nets with signed
    links to signed graphs (\cref{prop:sgn-petri-to-sgn-graph}).}
  \label{fig:petri-queries}
\end{figure}
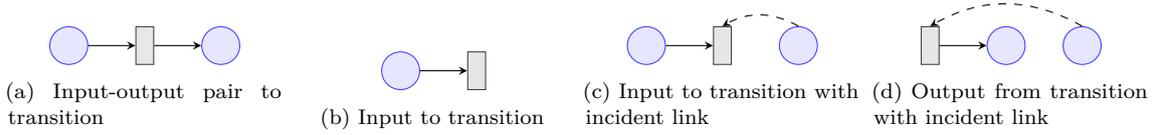

The resulting functor $\SgnPetri \to \SgnGraph$, which restricts on vertices to
the functor picking out the set of species of the Petri net (representable by
the Petri net with one species and no transitions) and restricts on edges to the
disjoint union of the four representable functors described above, has the form
of a \emph{familially representable functor} between copresheaf categories. This
type of functor has been studied extensively within both pure (see, for
instance, \cite[\mbox{Appendix C}]{leinster2004}) and applied category theory
(see \cite{spivak2021}). To be consistent with our existing notation and to
emphasize that such functors have a useful computational interpretation, we
adopt Spivak's perspective that they are \emph{data migration functors}. In this
framing, the schema category of a copresheaf category is regarded as the schema
for a database, a copresheaf on that category is regarded as a database instance
on the schema, and representable functors from the copresheaf category to $\Set$
are regarded as \emph{conjunctive queries}. A data migration functor between
copresheaf categories is then one which restricts on each component of the
codomain to a disjoint union of conjunctive queries, or \emph{duc-query}.

In order to apply this theory and precisely define the functor
$\SgnPetri \to \SgnGraph$, we fully schematize the definitions of signed graphs
and Petri nets with signed links. The \define{schema for signed graphs} is the
category $\Sch{\SgnGraph}$ freely generated by these objects and morphisms:
\begin{equation*}
  \begin{tikzcd}
    V & E & A \arrow[loop right, "\mathrm{neg}", looseness=7]
    \arrow["{\mathrm{src}}"', shift right=1, from=1-2, to=1-1]
    \arrow["{\mathrm{tgt}}", shift left=1, from=1-2, to=1-1]
    \arrow["{\mathrm{sgn}}", from=1-2, to=1-3]
  \end{tikzcd}.
\end{equation*}
A signed graph as in \cref{def:sgn-graph} is equivalent to a functor
$X: \Sch{\SgnGraph} \to \Set$ such that $X(A) = \Sgn$, the set of signs, and
$X(\negate): \Sgn \to \Sgn$ is negation (i.e., multiplication by $-1$). Note
that negation is not needed to define the data of a signed graph but is relevant
to the data migration. A morphism of signed graphs $X \to Y$ is a natural
transformation $\phi: X \to Y$ whose component at $A$ is the identity function,
$\phi_A = 1_\Sgn$. We thus obtain a category isomorphic to $\SgnGraph$.

Similarly, the \define{schema for Petri nets with signed links} is the category
$\Sch{\SgnPetri}$ freely generated by:
\begin{equation*}
  \begin{tikzcd}
    & S \\
    I & O & L & A \arrow[loop right, "\mathrm{neg}", looseness=7] \\
    & T
    \arrow["{\mathrm{src}_L}"', from=2-3, to=1-2]
    \arrow["{\mathrm{tgt}_L}", from=2-3, to=3-2]
    \arrow["{\mathrm{src}_I}", from=2-1, to=1-2]
    \arrow["{\mathrm{tgt}_I}"', from=2-1, to=3-2]
    \arrow["{\mathrm{tgt}_O}"{pos=0.3}, from=2-2, to=1-2]
    \arrow["{\mathrm{src}_O}"'{pos=0.3}, from=2-2, to=3-2]
    \arrow["{\mathrm{sgn}}", from=2-3, to=2-4]
    \arrow["{\mathrm{one}}"', curve={height=18pt}, from=3-2, to=2-4]
  \end{tikzcd}.
\end{equation*}
A Petri net with signed links, as in \cref{def:sgn-petri}, is equivalent to a
functor $P: \Sch{\SgnPetri} \to \Set$ such that $P(A) = \Sgn$, the map
$P(\negate): \Sgn \to \Sgn$ is negation, and $P(\one): P(T) \to \Sgn$ is the
constant map at $+1$. Again, these maps are needed for data migration, not for
the data itself. A morphism of Petri nets with signed links is a natural
transformation $\phi: P \to Q$ such $\phi_A = 1_\Sgn$, yielding a category
isomorphic to $\SgnPetri$.

With these preliminaries, we can construct the ``interactions functor,'' called
``$\Int$'' for short, from Petri nets with signed links to signed graphs.

\begin{proposition}[Interactions functor]
  \label{prop:sgn-petri-to-sgn-graph}
  A functor
  \begin{equation*}
    \Int: \SgnPetri \to \SgnGraph
  \end{equation*}
  is specified by the following functor from $\Sch{\SgnGraph}$ to the category
  of duc-queries on $\Sch{\SgnPetri}$:
  \begin{equation} \label{eq:sgn-petri-to-sgn-graph}
    \begin{aligned}
      \Sch{\SgnGraph} &\to \amalg\left(\left(\Set^{\Sch{\SgnPetri}}\right)^{\op}\right)\\
      V &\mapsto S \\
      E &\mapsto I \times_T O + I + I \times_T L + O \times_T L \\
      A &\mapsto A \\
      \src &\mapsto \left[
        \src_I \circ \pi_I,\ \src_I,\ \src_L \circ \pi_L,\ \src_L \circ \pi_L
      \right] \\
      \tgt &\mapsto \left[
        \tgt_O \circ \pi_O,\ \src_I,\ \src_I \circ \pi_I,\ \tgt_O \circ \pi_O
      \right] \\
      \sgn &\mapsto \left[
        \one \circ \pi_T,\ \negate \circ \one \circ \tgt_I,\
        \negate \circ \sgn \circ \pi_L,\ \sgn \circ \pi_L
      \right] \\
      \negate &\mapsto \negate.
    \end{aligned}
  \end{equation}
\end{proposition}
\begin{proof}
  We will define the functor $\Int: \SgnPetri \to \SgnGraph$ as the restriction of
  a functor $\Set^{\cat{D}} \to \Set^{\cat{C}}$ between the categories of
  copresheaves on $\cat{D} \coloneqq \Sch{\SgnPetri}$ and
  $\cat{C} \coloneqq \Sch{\SgnGraph}$. In fact, the functor
  $\Set^{\cat{D}} \to \Set^{\cat{C}}$ is of the special kind known as a
  \emph{parametric right adjoint} \cite[Definition p.\ 311]{street2000} or a
  \emph{familially representable functor} \cite[Appendix C]{leinster2004}.

  According to the theory of data migration \cite[Corollary 2.3.6]{spivak2021},
  giving a parametric right adjoint $\Set^{\cat{D}} \to \Set^{\cat{C}}$ is
  equivalent to giving a functor from $\cat{C}$ to
  $\amalg((\Set^{\cat{D}})^\op)$, the free coproduct completion of the free
  limit completion of $\cat{D}$. Our functor
  $\cat{C} \to \amalg((\Set^{\cat{D}})^\op)$ is defined by
  \cref{eq:sgn-petri-to-sgn-graph}. The assignment of $E \in \cat{C}$ can also
  be described as the sum of the four representables associated with the Petri
  nets with links in \cref{fig:petri-queries}.
  
  Finally, the assignments $A \mapsto A$ and $\negate \mapsto \negate$ ensure
  that if $P$ is a copresheaf on $\cat{D}$ with $P(A) = \Sgn$ and $P(\negate)$
  is negation, then applying this parametric right adjoint functor to $P$ yields
  a copresheaf $X$ on $\cat{C}$ where again $X(A) = \Sgn$ and $X(\negate)$ is
  negation. Thus, this functor between copresheaf categories restricts to a
  functor $\SgnPetri \to \SgnGraph$ as claimed.
\end{proof}

Using the interactions functor, we can give a formal account of what it means to
have a mechanistic model for a regulatory network.

\begin{definition}[Mechanism]
  A \define{mechanistic model} for a regulatory network $X$ is a Petri net with
  signed links $P$ together with an occurrence of $X$ in $\Int(P)$, i.e., a
  monic signed functor $X \monicto \Int(P)$.
\end{definition}

For example, the Petri net with signed links in \cref{eq:voit-ex-sgn-petri} is a
mechanistic model for a regulatory network in which $A$ and $D$ participate in a
positive feedback loop:
\begin{equation*}
  \begin{tikzcd}
    A & D
    \arrow[curve={height=6pt}, from=1-1, to=1-2]
    \arrow[curve={height=6pt}, from=1-2, to=1-1]
  \end{tikzcd}.
\end{equation*}

The next example is inspired by activation and inhibition of enzymes
\cite[Chapter 3]{ingalls2013}. Besides being of greater biological interest, it
shows why Petri nets with links are better suited than bare Petri nets to
describe mechanisms.

\begin{example}[Enzyme inhibition]
  Suppose that an enzyme $E$ catalyzes the conversion of a substrate $S$ into a
  product $P$, and that an inhibiting agent $I$ transforms the chemical
  configuration of the reactive enzyme $E$ into the inert compound $F$. Then the
  agent $I$ inhibits the product $P$:
  \begin{equation*}
    \begin{tikzcd}
      I & P
      \arrow[maps to, no head, from=1-2, to=1-1]
    \end{tikzcd}.
  \end{equation*}
  As a first attempt to model these interactions, consider the following Petri
  net without links.
  
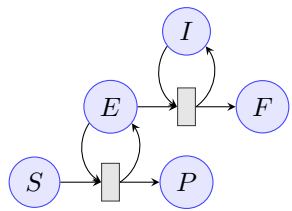
\begin{equation*}
    \begin{tikzpicture}[baseline={(current bounding box.center)}]
      \node[species] (I) at (0,1) {$I$};
      \node[species] (E) at (-1,0) {$E$};
      \node[species] (S) at (-2,-1) {$S$};
      \node[transition] (in1) at (0,0) {};
      \node[transition] (in2) at (-1,-1) {};
      \node[species] (F) at (1,0) {$F$};
      \node[species] (P) at (0,-1) {$P$};
      \draw[arc] (in1) to (F);
      \draw[arc, bend right=45] (in1.east) to (I);
      \draw[arc] (in2) to (P);
      \draw[arc, bend right=45] (I) to (in1.west);
      \draw[arc] (E) to (in1);
      \draw[arc, bend right=45] (E) to (in2.west);
      \draw[arc] (S) to (in2);
      \draw[arc, bend right=45] (in2.east) to (E);
    \end{tikzpicture}
  \end{equation*}
  One might expect this Petri net to provide a mechanism for the inhibition of
  $P$ by $I$, but at least according to our definition, it does not. Applying
  the functor $\Int: \SgnPetri \to \SgnGraph$ yields the regulatory network
  
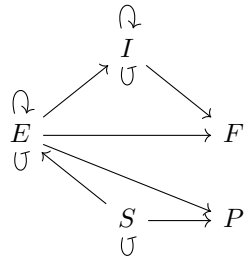
\begin{equation*}
    \begin{tikzcd}
      {} & I \arrow[loop below, maps to, no head, looseness=7] \arrow[loop above, looseness=7] & {}\\
      E \arrow[loop below, maps to, no head, looseness=7] \arrow[loop above, looseness=7] & {} & F \\
      {} & S \arrow[loop below, maps to, no head, looseness=7] & P
      \arrow[from=1-2, to=2-3]
      \arrow[from=2-1, to=1-2]
      \arrow[from=2-1, to=2-3]
      \arrow[from=2-1, to=3-3]
      \arrow[from=3-2, to=2-1]
      \arrow[from=3-2, to=3-3]
    \end{tikzcd}
  \end{equation*}
  in which there is no occurrence of the desired inhibition. If instead we allow
  links, the reaction catalyzed by $E$ should be written as having input $S$,
  output $P$, and a link from $E$ to the transition. A similar change is
  required for the reaction involving $E$, $F$, and $I$. A better model of the
  reactions is thus the following Petri net with links.
  
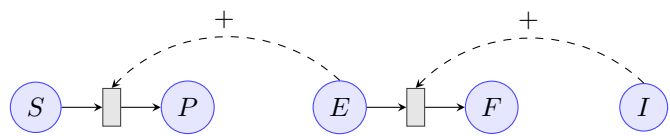
\begin{equation*}
    \begin{tikzpicture}[baseline={(current bounding box.center)}]
      \node[species] (I) at (3,0) {$I$};
      \node[species] (E) at (-1,0) {$E$};
      \node[species] (S) at (-5,0) {$S$};
      \node[transition] (in1) at (0,0) {};
      \node[transition] (in2) at (-4,0) {};
      \node[species] (F) at (1,0) {$F$};
      \node[species] (P) at (-3,0) {$P$};
      \draw[arc] (in1) to (F);
      \draw[arc] (in2) to (P);
      \draw[arc] (E) to (in1);
      \draw[arc] (S) to (in2);
      \draw[link, bend right=45] (E.north) to node[midway,above] {$+$} (in2.north);
      \draw[link, bend right=45] (I.north) to node[midway,above] {$+$} (in1.north);
    \end{tikzpicture}
  \end{equation*}
  Applying the functor $\Int$ now yields the network
  
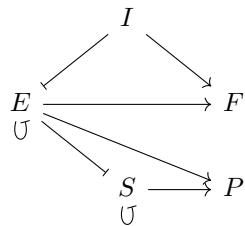
\begin{equation*}
    \begin{tikzcd}
      {} & I & {}\\
      E \arrow[loop below, maps to, no head, looseness=7] & {} & F \\
      {} & S \arrow[loop below, maps to, no head, looseness=7] & P
      \arrow[from=1-2, to=2-3]
      \arrow[from=2-1, to=2-3]
      \arrow[from=2-1, to=3-3]
      \arrow[from=3-2, to=3-3]
      \arrow[maps to, no head, from=2-1, to=1-2]
      \arrow[maps to, no head, from=3-2, to=2-1]
    \end{tikzcd}
  \end{equation*}
  which indeed has an occurrence of $I$ inhibiting $P$.
\end{example}

\section{Quantitative analysis: parameters and dynamics}
\label{sec:quantitative}

\subsection{Parameterized dynamical systems}
\label{sec:para-dynam}

Pioneering the idea of functorial semantics for scientific models, Baez and
Pollard extended the mass-action model of reaction networks to a functor from
the category of Petri nets with rates into a category of dynamical systems
\cite{baez2017}. In this picture, the reaction rate coefficients are known
constants associated with the reaction network. In practice, however, rate
coefficients are often unknown and must be extracted from existing literature or
estimated from experimental data. We therefore change our perspective slightly
and consider dynamical systems not in isolation but as \emph{parameterized
  families}. This shift also turns out to have formal advantages: the category
of parameterized dynamical systems is better behaved than the category of
dynamical systems, which has too few morphisms.

To begin, we recall the dynamics functor, nearly identical to Baez-Pollard's
\cite[Lemma 15]{baez2017}:

\begin{lemma}[Dynamics] \label{lem:dynamics}
  There is a functor $\Dynam: \FinSet \to \Vect$ that sends
  \begin{itemize}
    \item a finite set $S$ to the vector space of algebraic vector fields
      $v: \R^S \to \R^S$, where \define{algebraic} means that the components of
      the vector field are polynomials in the state variables;
    \item a function $f: S \to S'$ between finite sets to the linear
      transformation
      \begin{equation*}
        (v: \R^S \to \R^S) \quad\mapsto\quad
        (f_* \circ v \circ f^*: \R^{S'} \to \R^{S'}),
      \end{equation*}
      where the linear map $f^*: \R^{S'} \to \R^S$ is the \define{pullback}
      along $f$
      \begin{equation*}
        f^*(x')(i) \coloneqq x'(f(i)),
        \qquad x' \in \R^{S'},\ i \in S,
      \end{equation*}
      and the linear map $f_*: \R^S \to \R^{S'}$ is the \define{pushforward} along $f$
      \begin{equation*}
        f_*(x)(i') \coloneqq \sum_{i \in f^{-1}(i')} x(i),
        \qquad x \in \R^S,\ i' \in S'.
      \end{equation*}
  \end{itemize}
\end{lemma}
\begin{proof}
  The functor $\Dynam: \FinSet \to \Vect$ can be constructed as the composite
  functor
  \begin{equation*}
    \FinSet \xrightarrow{\langle D, F \rangle}
      \Vect^\op \times \Vect \xrightarrow{\Poly(-,-)}
      \Vect,
  \end{equation*}
  where $F: \FinSet \to \Vect,\ X \mapsto \R^X$ is the free vector space functor
  (restricted to finite sets); $D: \FinSet^\op \to \Vect$ is the dual vector
  space functor (restricted to $F$), whose underlying set-valued functor is
  \begin{equation*}
    \Vect(F(-), \R) \cong \Set(-, \R): \FinSet^\op \to \Set;
  \end{equation*}
  and $\Poly(-,-)$ is the $\Vect$-valued hom-functor that sends a pair of vector
  spaces to the vector space of polynomial functions between them.\footnote{For
    a coordinate-free description of polynomial maps between vector spaces, see
    \cite[\S 1.6]{cartan1971}.}
\end{proof}

The dynamics functor is the same one studied by Baez and Pollard except that we
take the vector space, rather than merely the set, of vector fields. That is
because we are interested in \emph{linearly} parameterized dynamical systems. In
calling the functor ``dynamics,'' we implictly identify a vector field with the
differentiable dynamical system that it generates. This common practice is not
entirely innocent since even when a system of differential equations depends
linearly on parameters, its solutions rarely do. We also note that the
restriction to \emph{algebraic} vector fields, as opposed to smooth or even just
continuous ones, is inessential but suffices for us and agrees with
Baez-Pollard.

The dynamics functor is the main building block in constructing a category of
parameterized dynamical systems.

\begin{definition}[Linear parameterizations] \label{def:para-dynam}
  The \define{category of linearly parameterized dynamical systems} is the comma
  category
  \begin{equation*}
    \LinParaDynam \coloneqq F/\Dynam,
  \end{equation*}
  where $F: \FinSet \to \Vect,\ X \mapsto \R^X$ is the free vector space functor
  restricted to finite sets.
\end{definition}

So, by definition, a linearly parameterized dynamical system consists of a
finite set $P$, the \define{parameter variables}, and a finite set $S$, the
\define{state variables}, together with a linear map
\begin{equation*}
  v: \R^P \to \Dynam(S)
\end{equation*}
sending each choice of parameters $\theta \in \R^P$ to an algebraic vector field
$v(\theta)$. In more conventional notation, we can write
$v(x; \theta) \coloneqq v(\theta)(x)$ for $x \in \R^S$ and $\theta \in \R^P$. A
morphism $(P,S,v) \to (P',S',v')$ of linearly parameterized dynamical systems is
a pair of functions $q: P \to P'$ and $f: S \to S'$ making the square
\begin{equation} \label{eq:para-dynam-hom}
  \begin{tikzcd}
    {\R^P} & {\Dynam(S)} \\
    {\R^{P'}} & {\Dynam(S')}
    \arrow["v", from=1-1, to=1-2]
    \arrow["{f_* \circ (-) \circ f^*}", from=1-2, to=2-2]
    \arrow["{v'}"', from=2-1, to=2-2]
    \arrow["{q_*}"', from=1-1, to=2-1]
  \end{tikzcd}
\end{equation}
commute, or equivalently making the equation
\begin{equation*}
  f_*(v(f^*(x'); \theta)) = v'(x'; q_*(\theta))
\end{equation*}
hold for all $x' \in \R^{S'}$ and $\theta \in \R^P$.

While certainly not all dynamical models depend linearly on their parameters, a
great many of them do, including several important canonical models in biology
and chemistry. The law of mass action defines dynamical systems that depend
linearly on the rate coefficients. The generalized Lotka-Volterra equations,
studied in the next section, are linear in the rate and affinity parameters. Of
course, the mass-action and Lotka-Volterra equations are nonlinear ODEs;
linearity of a vector field in state or in parameters are separate matters.
Nevertheless, even for nonlinear models such as Lotka-Volterra, linearity in
parameters is a useful assumption that aides in the identifiability analysis of
the model \cite[\S 5]{stanhope2014}.

To express important physical constraints and to define a semantics for signed
graphs, we will restrict the dynamical system and its parameters to be
nonnegative. This is straightforward enough but requires a bit of additional
formalism.

Write $\R_+ \coloneqq \{x \in \R: x \geq 0\}$ for the semiring of nonnegative
real numbers. A module over $\R_+$ is called a \define{conical space}, and the
category of conical spaces and conic-linear ($\R_+$-linear) maps is denoted
$\Con \coloneqq \Mod{\R_+}$. A conical space is a structure in which one can
take linear combinations with nonnegative real coefficients, just as a real
vector space ($\R$-module) is a structure in which one can take linear
combinations with arbitrary real coefficients. Any convex cone in a real vector
space is a conical space. Our main example is the \define{nonnegative orthant}
of $\R^S$ for some set $S$: the function space
$\R_+^S \coloneqq \{x: S \to \R_+\}$, with conical combinations taken pointwise.
A real vector space can itself be regarded as a conical space; more precisely,
the inclusion of semirings $\R_+ \hookrightarrow \R$ induces a forgetful functor
$\Vect \to \Con$ by restriction of scalars.

Recall that a dynamical system is \define{nonnegative} if whenever the initial
condition is in the nonnegative orthant, its trajectory always remains in the
nonnegative orthant. A dynamical system of form $\dot x = v(x)$ is nonnegative
if and only if $v_i(x) \geq 0$ whenever $x \geq 0$ componentwise and $x_i = 0$
\cite[Proposition 2.1]{haddad2010}, in which case the vector field $v$ is called
\define{essentially nonnegative} \cite[Definition 2.1]{haddad2010}. Using this
criterion, it is easy to see that a reaction network with mass-action kinetics
is nonnegative assuming the rate constants are nonnegative, as is a
Lotka-Volterra system for any choice of parameters. Hence both systems satisfy
the obvious physical constraint that no species should have negative
concentration or population.

\begin{lemma}[Nonnegative dynamics] \label{lem:nonneg-dynam}
  There is a functor $\Dynam_+: \FinSet \to \Con$ that sends a finite set $S$ to
  the conical space of essentially nonnegative, algebraic vector fields
  $v: \R^S \to \R^S$ and sends a function $f: S \to S'$ to the transformation
  $v \mapsto f_* \circ v \circ f^*$.
\end{lemma}
\begin{proof}
  The proof is similar to that of \cref{lem:dynamics}. It is clear that the
  essentially nonnegative functions are stable under pointwise conical
  combinations, hence form a conical space. (They are, of course, \emph{not}
  stable under arbitrary linear combinations.) We just need to check that if
  $v: \R^S \to \R^S$ is essentially nonnegative, then so is the transformed
  vector field $f_* \circ v \circ f^*: \R^{S'} \to \R^{S'}$. Fix
  $x' \in \R_+^{S'}$ and $i' \in S'$, and suppose that $x'(i') = 0$. For every
  $i \in f^{-1}(i')$, we have $f^*(x')(i) = x'(f(i)) = x'(i') = 0$ and so
  $v(f^*(x'))(i) \geq 0$, whence the inequality of essential nonnegativity
  follows:
  \begin{equation*}
    (f_* \circ v \circ f^*)(x')(i')
      = \sum_{i \in f^{-1}(i')} v(f^*(x'))(i)
      \geq 0. \qedhere
  \end{equation*}
\end{proof}

We now define the conical analogue of linearly parameterized dynamical systems.

\begin{definition}[Conical parameterizations] \label{def:para-nonneg-dynam}
  The \define{category of conically parameterized nonnegative dynamical systems}
  is the comma category
  \begin{equation*}
    \ConParaDynam \coloneqq F_+ / \Dynam_+.
  \end{equation*}
  where $F_+: \FinSet \to \Con$, $X \mapsto \R_+^X$ is the free conical space
  functor restricted to finite sets.
\end{definition}

So, a conically parameterized nonnegative dynamical system consists of finite
sets $P$ and $S$ together with a conic-linear map
\begin{equation*}
  v: \R_+^P \to \Dynam_+(S).
\end{equation*}

\begin{proposition}[Colimits of parameterized dynamical systems]
  \label{prop:colim-para-dynam}
  The categories of linearly and conically parameterized dynamical systems are
  finitely cocomplete. Moreover, these finite colimits are computed by colimits
  in $\FinSet$ of the parameter variables and of the state variables.
\end{proposition}
\begin{proof}
  The category $\FinSet$ has finite colimits and the functors
  $F: \FinSet \to \Vect$ and $F_+: \FinSet \to \Con$ preserve finite colimits,
  since they are composites of the inclusion $\FinSet \hookrightarrow \Set$ with
  the left adjoints
  \begin{equation*}
    \begin{tikzcd}
      \Set & \Vect
      \arrow[""{name=0, anchor=center, inner sep=0}, "F", curve={height=-12pt}, from=1-1, to=1-2]
      \arrow[""{name=1, anchor=center, inner sep=0}, "U", curve={height=-12pt}, from=1-2, to=1-1]
      \arrow["\dashv"{anchor=center, rotate=-90}, draw=none, from=0, to=1]
    \end{tikzcd}
    \qquad\text{and}\qquad
    \begin{tikzcd}
      \Set & \Con
      \arrow[""{name=0, anchor=center, inner sep=0}, "{F_+}", curve={height=-12pt}, from=1-1, to=1-2]
      \arrow[""{name=1, anchor=center, inner sep=0}, "{U_+}", curve={height=-12pt}, from=1-2, to=1-1]
      \arrow["\dashv"{anchor=center, rotate=-90}, draw=none, from=0, to=1]
    \end{tikzcd}
  \end{equation*}
  to the underlying set functors on vector spaces and conical spaces. By
  \cref{lem:colim-comma-cat} below, the comma categories
  $\LinParaDynam = F/\Dynam$ and $\ConParaDynam = F_+/\Dynam_+$ have finite
  colimits, which are preserved by the projection functors onto $\FinSet$.
\end{proof}

To illustrate, we describe the initial object and binary coproducts in
$\LinParaDynam$. The initial linearly parameterized dynamical system has no
parameter variables, no state variables, and the unique (trivial) vector field
on the zero vector space. The coproduct of two linearly parameterized dynamical
systems $(P_1,S_1,v_1)$ and $(P_2,S_2,v_2)$ has parameter variables $P_1 + P_2$,
state variables $S_1 + S_2$, and parameterized vector field
\begin{equation*}
  \R^{P_1+P_2} \cong \R^{P_1} \oplus \R^{P_2}
    \xrightarrow{v_1 \oplus v_2}
    \Dynam(S_1) \oplus \Dynam(S_2)
    \xrightarrow{[\Dynam(\iota_1), \Dynam(\iota_2))]}
    \Dynam(S_1 + S_2),
\end{equation*}
where $\iota_1: S_1 \to S_1 + S_2$ and $\iota_2: S_2 \to S_1 + S_2$ are the
canonical inclusions. In conventional notation, the coproduct system has
parameterized vector field
\begin{equation*}
  v\left(\begin{bmatrix} x_1 \\ x_2 \end{bmatrix};
         \begin{bmatrix} \theta_1 \\ \theta_2 \end{bmatrix}\right) =
  \begin{bmatrix}
    v_1(x_1; \theta_1) \\
    v_2(x_2; \theta_2)
  \end{bmatrix}.
\end{equation*}

The proof of \cref{prop:colim-para-dynam}, as well as of
\cref{thm:lv-graph,thm:lv-sgn-graph} below, depends on the following technical
lemma about comma categories, which the reader can omit without loss of
continuity.

\begin{lemma}[Colimits in comma categories] \label{lem:colim-comma-cat}
  Let $\cat{C}_0 \xrightarrow{F_0} \cat{C} \xleftarrow{F_1} \cat{C}_1$ be a
  cospan of categories such that $\cat{C}_0$ and $\cat{C}_1$ have colimits of
  shape $\cat{J}$ and $F_0$ preserves $\cat{J}$-shaped colimits. Then the comma
  category $F_0/F_1$ also has $\cat{J}$-shaped colimits, and the projection
  functors $\pi_i: F_0/F_1 \to \cat{C}_i$, $i=0,1$, preserve those colimits.

  Furthermore, a functor $G: \cat{X} \to F_0/F_1$ into the comma category
  preserves $\cat{J}$-shaped colimits whenever the associated functors
  $G_i \coloneqq \pi_i \circ G: \cat{X} \to \cat{C}_i$, $i = 0,1$, do so.
\end{lemma}
\begin{proof}
  Colimits in the comma category $F_0/F_1$ are constructed in \cite[\S
  5.2]{rydeheard1988}. To make the rest of the proof self-contained, we recall
  the construction here.

  By the universal property of the comma category, a diagram
  $D: \cat{J} \to F_0/F_1$ is equivalent to diagrams
  $D_i \coloneqq \pi_i \circ D: \cat{J} \to \cat{C}_i$, $i=0,1$, along with a
  natural transformation $\vec D: F_0 \circ D_0 \To F_1 \circ D_1$. Let
  $(c_i, \lambda^i)$ be a colimit cocone for the diagram $D_i$ in $\cat{C}_i$,
  having legs $D_i(j) \xrightarrow{\lambda^i_j} c_i$ for each $j \in \cat{J}$.
  The family of morphisms
  \begin{equation*}
    F_0(D_0(j)) \xrightarrow{\vec D_j} F_1(D_1(j))
      \xrightarrow{F_1 (\lambda^1_j)} F_1(c_1),
    \qquad j \in \cat{J},
  \end{equation*}
  is then a cocone under $F_0 \circ D_0$. Since $F_0$ preserves $\cat{J}$-shaped
  limits, $(F_0(c_0), F_0 * \lambda^0)$ is a colimit cocone for $F_0 \circ D_0$,
  so by its universal property, there exists a unique morphism
  $f: F_0(c_0) \to F_1(c_1)$ making the squares commute:
  \begin{equation*}
    \begin{tikzcd}
      {F_0(D_0(j))} & {F_1(D_1(j))} \\
      {F_0(c_0)} & {F_1(c_1)}
      \arrow["{\vec D_j}", from=1-1, to=1-2]
      \arrow["f"', from=2-1, to=2-2]
      \arrow["{F_0(\lambda^0_j)}"', from=1-1, to=2-1]
      \arrow["{F_1(\lambda^1_j)}", from=1-2, to=2-2]
    \end{tikzcd},
    \qquad j \in \cat{J}.
  \end{equation*}
  Setting $\lambda \coloneqq (\lambda_j^0, \lambda_j^1)_{j \in \cat{J}}$, the
  cocone $((c_0, c_1, f), \lambda)$ can be shown to be a colimit of the diagram
  $D$.

  To prove the last statement about colimit preservation, let
  $D: \cat{J} \to \cat{X}$ be a diagram with colimit cocone $(x, \lambda)$,
  having legs $Dj \xrightarrow{\lambda_j} y$ for $j \in \cat{J}$. We must show
  that its image cocone $(G(x), G*\lambda)$ is a colimit of the diagram
  $G \circ D$ in $F_0/F_1$. By the universal property of the comma category, the
  functor $G: \cat{X} \to F_0/F_1$ is equivalent to the functors
  $G_i: \cat{X} \to \cat{C}_i$, $i=0,1$, along with a natural transformation
  $\vec G: F_0 \circ G_0 \To F_1 \circ G_1$. The image cocone
  $(G(x), G*\lambda)$ then amounts to cocones $(G_0(x), G_0 * \lambda)$ and
  $(G_1(x), G_1 * \lambda)$, which by hypothesis are colimits of the diagrams
  $G_0 \circ D$ and $G_1 \circ D$ in $\cat{C}_0$ and $\cat{C}_1$, together with
  a family of commutative squares in $\cat{C}$:
  \begin{equation*}
    \begin{tikzcd}
      {F_0(G_0(Dj))} & {F_1(G_1(Dj))} \\
      {F_0(G_0(x))} & {F_1(G_1(x))}
      \arrow["{\vec G_{Dj}}", from=1-1, to=1-2]
      \arrow["{F_0(G_0(\lambda_j))}"', from=1-1, to=2-1]
      \arrow["{\vec G_x}"', from=2-1, to=2-2]
      \arrow["{F_1(G_1(\lambda_j))}", from=1-2, to=2-2]
    \end{tikzcd},
    \qquad j \in \cat{J}.
  \end{equation*}
  But a morphism $\vec G_x$ making all these squares commute is already uniquely
  determined by the universal property of the colimit cocone
  $(F_0(G_0(x)), F_0 * G_0 * \lambda)$ for the diagram $F_0 \circ G_0 \circ D$,
  using the hypothesis that $F_0$ preserves $\cat{J}$-shaped colimits. Indeed,
  this is precisely how one constructs the colimit of the diagram $G \circ D$ in
  $F_0/F_1$, as shown above. It follows that $(G(x), G*\lambda)$ is a colimit
  cocone for $G \circ D$.
\end{proof}

\subsection{The Lotka-Volterra dynamical model}
\label{sec:lv}

A \define{Lotka-Volterra system} with $n$ species has, using matrix notation,
the vector field
\begin{equation*}
  v(x; \rho, \beta) \coloneqq x \odot (\rho + \beta x) = \diag(x) (\rho + \beta x)
\end{equation*}
with state vector $x \in \R^n$ and arbitrary real-valued parameters
$\rho \in \R^n$ and $\beta \in \R^{n \times n}$ \cite[\S 2.2]{szederkenyi2018}.
In coordinates, the vector field is
\begin{equation*}
  v_i(x; \rho, \beta)
    = x_i \left(\rho_i + \sum_{j=1}^n \beta_{i,j} x_j \right)
    = \rho_i x_i + \sum_{j=1}^n \beta_{i,j} x_i x_j,
  \qquad i = 1,\dots,n.
\end{equation*}
The parameter $\rho_i$ sets the baseline rate of growth (when positive) or decay
(when negative) for species $i$, whereas $\beta_{i,j}$ defines a promoting (when
positive) or inhibiting (when negative) effect of species $j$ on species $i$. In
typical applications the signs of the parameters are fixed and known in advance
of any data. For example, in the famous \define{predator-prey Lotka-Volterra
  system}
\begin{align*}
  \dot x &= ax - bxy \\
  \dot y &= dxy - cy,
\end{align*}
with prey $x$ and predators $y$, the parameters
$\rho = \begin{bmatrix} a \\ -c \end{bmatrix}$ and
$\beta = \begin{bmatrix} 0 & -b \\ d & 0 \end{bmatrix}$ are specified by
nonnegative real numbers $a, b, c, d \geq 0$.

In this section, we define quantitative semantics for graphs and signed graphs
using the Lotka-Volterra dynamical model. To illustrate the main ideas, we first
construct a functor from finite graphs (\cref{def:graph}) to linearly
parameterized dynamical systems (\cref{def:para-dynam}), giving a semantics for
unlabeled graphs. It is more useful to have a semantics for regulatory networks,
which we have defined to be signed graphs. We therefore construct a second
functor from finite signed graphs (\cref{def:sgn-graph}) to conically
parameterized nonnegative dynamical systems (\cref{def:para-nonneg-dynam}).

Recall that a graph is \define{finite} if its vertex and edge sets are both
finite. Let $\FinGraph$ denote the full subcategory of $\Graph$ spanned by
finite graphs.

\begin{theorem}[Lotka-Volterra model for finite graphs] \label{thm:lv-graph}
  There is a functor
  \begin{equation*}
    \LV: \FinGraph \to \LinParaDynam
  \end{equation*}
  that sends
  \begin{itemize}[noitemsep]
    \item a finite graph $X$ to the linearly parameterized dynamical system with
      parameter variables $P \coloneqq X(V) + X(E)$, state variables
      $S \coloneqq X(V)$, and algebraic vector field\footnote{For a fixed graph
      $X$ and vertex $i \in X(V)$, the notation $(e: i' \to i) \in X$ means any
      edge $e \in X(\tgt)^{-1}(i)$ incoming to $i$, whose source
      $i' = X(\src)(e)$ varies with $e$.}
      \begin{equation*}
        v(x; \rho, \beta)(i) \coloneqq \rho(i)\, x(i) +
          \sum_{(e: i' \to i) \in X} \beta(e)\, x(i')\, x(i),
        \qquad x \in \R^{X(V)},\ i \in X(V),
      \end{equation*}
      parameterized by vectors $\rho \in \R^{X(V)}$ and
      $\beta \in \R^{X(E)}$;
    \item a graph homomorphism $\phi: X \to Y$ to a morphism of systems with
      parameter variable map $\phi_V + \phi_E: X(V) + X(E) \to Y(V) + Y(E)$ and
      state variable map $\phi_V: X(V) \to Y(V)$.
  \end{itemize}
  Moreover, the functor $\LV$ preserves finite colimits.
\end{theorem}
\begin{proof}
  By the universal property of the comma category $\LinParaDynam = F/\Dynam$, to
  give a functor $\LV: \FinGraph \to \LinParaDynam$ is to give a pair of
  functors $\LV_0, \LV_1: \FinGraph \to \FinSet$ along with a natural
  transformation
  \begin{equation*}
    \vec\LV: (F \circ \LV_0) \To (\Dynam \circ \LV_1): \FinGraph \to \Vect.
  \end{equation*}
  We set $\LV_0(X) \coloneqq X(V) + X(E)$ and $\LV_1(X) \coloneqq X(V)$. Using
  the universal property of the coproduct in $\Vect$, the components
  \begin{equation*}
    \vec\LV_X: \R^{X(V)} \oplus \R^{X(E)} \cong \R^{X(V) + X(E)} \to \Dynam(X(V)).
  \end{equation*}
  of the transformation $\vec\LV$ themselves decompose into two parts, call them
  \begin{equation*}
    v_X^0 \coloneqq \vec\LV_X^0: \R^{X(V)} \to \Dynam(X(V))
    \quad\text{and}\quad
    v_X^1 \coloneqq \vec\LV_X^1: \R^{X(E)} \to \Dynam(X(V)).
  \end{equation*}
  We define these to be
  \begin{equation*}
    v_X^0(x; \rho)(i) \coloneqq \rho(i)\, x(i)
    \qquad\text{and}\qquad
    v_X^1(x; \beta)(i) \coloneqq \sum_{e \in X(\tgt)^{-1}(i)} \beta(e)\, x(X(\src)(e))\, x(i).
  \end{equation*}
  Putting the pieces back together reproduces the first statement of the
  theorem. We just need to check that the transformation $\vec\LV$ is, in fact,
  natural.

  Given a graph homomorphism $\phi: X \to Y$, the naturality square for the
  transformation $\vec\LV$ is
  \begin{equation} \label{eq:lv-naturality}
    \begin{tikzcd}
      {\R^{X(V)+X(E)}} & {\Dynam(X(V))} \\
      {\R^{Y(V)+Y(E)}} & {\Dynam(Y(V))}
      \arrow["{\vec\LV_X}", from=1-1, to=1-2]
      \arrow["{(\phi_V + \phi_E)_*}"', from=1-1, to=2-1]
      \arrow["{(\phi_V)_* \circ (-) \circ (\phi_V)^*}", from=1-2, to=2-2]
      \arrow["{\vec\LV_Y}"', from=2-1, to=2-2]
    \end{tikzcd},
  \end{equation}
  which decomposes into two squares,
  \begin{equation*}
    \begin{tikzcd}
      {\R^{X(V)}} & {\Dynam(X(V))} \\
      {\R^{Y(V)}} & {\Dynam(Y(V))}
      \arrow["{v_X^0}", from=1-1, to=1-2]
      \arrow["{(\phi_V)_*}"', from=1-1, to=2-1]
      \arrow["{(\phi_V)_* \circ (-) \circ (\phi_V)^*}", from=1-2, to=2-2]
      \arrow["{v_Y^0}"', from=2-1, to=2-2]
    \end{tikzcd}
    \qquad\text{and}\qquad
    \begin{tikzcd}
      {\R^{X(E)}} & {\Dynam(X(V))} \\
      {\R^{Y(E)}} & {\Dynam(Y(V))}
      \arrow["{v_X^1}", from=1-1, to=1-2]
      \arrow["{(\phi_E)_*}"', from=1-1, to=2-1]
      \arrow["{(\phi_V)_* \circ (-) \circ (\phi_V)^*}", from=1-2, to=2-2]
      \arrow["{v_Y^1}"', from=2-1, to=2-2]
    \end{tikzcd}.
  \end{equation*}
  Let us check that both squares commute. For the first, we have
  \begin{align*}
    (\phi_V)_*(v_X^0(\phi_V^*(y); \rho))(j)
      &= \sum_{i \in \phi_V^{-1}(j)} v_X^0(y \circ \phi_V; \rho)(i)
       = \sum_{i \in \phi_V^{-1}(j)} \rho(i)\, y(\phi_V(i)) \\
      &= \left(\sum_{i \in \phi_V^{-1}(j)} \rho(i) \right) y(j)
       = v_Y^0(y; (\phi_V)_*(\rho))(j)
  \end{align*}
  for all $y \in \R^{Y(V)}$, $\rho \in \R^{X(V)}$, and $j \in Y(V)$. For the
  second, we have
  \begin{align*}
    (\phi_V)_*(v_X^1(\phi_V^*(y); \beta))(j)
      &= \sum_{i \in \phi_V^{-1}(j)} v_X^1(y \circ \phi_V; \beta)(i) \\
      &= \sum_{i \in \phi_V^{-1}(j)} \sum_{e \in X(\tgt)^{-1}(i)}
         \beta(e)\, y(\phi_V(X(\src)(e)))\, y(j) \\
      &= \sum_{f \in Y(\tgt)^{-1}(j)} \sum_{e \in \phi_E^{-1}(f)}
         \beta(e)\, y(Y(\src)(\phi_E(e)))\, y(j) \\
      &= \sum_{f \in Y(\tgt)^{-1}(j)}
         \left(\sum_{e \in \phi_E^{-1}(f)} \beta(e)\right)
         y(Y(\src)(f)) y(j) \\
      &= v_Y^1(y; (\phi_E)_*(\beta))(j),
  \end{align*}
  for all $y \in \R^{Y(V)}$, $\beta \in \R^{X(E)}$, and $j \in Y(V)$. When
  exchanging the order of the summations we have used the facts that the graph
  homomorphism $\phi: X \to Y$ preserves sources and targets, the latter in its
  contravariant form
  \begin{equation*}
    \begin{tikzcd}
      {X(E)} & {X(V)} \\
      {Y(E)} & {Y(V)}
      \arrow["{X(\mathrm{tgt})}", from=1-1, to=1-2]
      \arrow["{\phi_E}"', from=1-1, to=2-1]
      \arrow["{Y(\mathrm{tgt})}"', from=2-1, to=2-2]
      \arrow["{\phi_V}", from=1-2, to=2-2]
    \end{tikzcd}
    \qquad\leadsto\qquad
    \begin{tikzcd}
      {\mathcal{P}(Y(V))} & {\mathcal{P}(X(V))} \\
      {\mathcal{P}(Y(E))} & {\mathcal{P}(X(E))}
      \arrow["{X(\mathrm{tgt})^{-1}}", from=1-2, to=2-2]
      \arrow["{\phi_E^{-1}}"', from=2-1, to=2-2]
      \arrow["{Y(\mathrm{tgt})^{-1}}"', from=1-1, to=2-1]
      \arrow["{\phi_V^{-1}}", from=1-1, to=1-2]
    \end{tikzcd},
  \end{equation*}
  where $\mathcal{P}(S)$ denotes the power set of a set $S$.

  Finally, we must verify that the functor $\LV$ preserves finite colimits. By
  \cref{lem:colim-comma-cat}, that happens provided both functors
  $\LV_0, \LV_1: \FinGraph \to \FinSet$ preserve finite colimits. The functor
  $\LV_1 = \ev_V$ is an evaluation functor on a copresheaf category, hence
  preserves colimits \cite[Proposition 3.3.9]{riehl2016}. Since coproducts
  commute with colimits, the pointwise coproduct of two evaluation functors
  \begin{equation*}
    \LV_0 = \big(\FinGraph
      \xrightarrow{\langle \ev_V, \ev_E \rangle} \FinSet \times \FinSet
      \xrightarrow{+} \FinSet\big)
  \end{equation*}
  also preserves colimits. This completes the proof.
\end{proof}

A quantitative semantics for signed graphs can be defined similarly, subject to
a caveat about the vertex parameters. Our notion of signed graph, designed to
capture regulatory networks as studied in the biochemistry literature, attaches
signs only to edges. We are thus led to assume that, in the Lotka-Volterra
dynamical model, all species have baseline rates of \emph{decay} rather than
growth. This assumption is valid for some, though certainly not all, regulatory
networks \cite[\mbox{p.\ 222}]{tyson2010}. It is not suitable for predator-prey
models in ecology.

Baseline rates of growth or decay could, if needed, be parameterized more
flexibly. Most simply, one could attach signs to vertices as well as edges and
use them in the quantitative semantics. Alternatively, at the expense of a more
cumbersome formalism, one could define dynamical systems with mixed
linear-conical parameterizations, allowing the vertex parameters to be arbitrary
real numbers while the edge parameters are constrained to be
nonnegative.\footnote{Similar mixed parameterizations are a practical necessity
  in parametric \emph{statistical} models, studied in detail in one author's PhD
  thesis \cite{patterson2020}.} For simplicity and uniformity of presentation,
we do not describe these extensions further.

Let $\FinSgnGraph$ denote the full subcategory of $\SgnGraph$ spanned by finite
signed graphs.

\begin{theorem}[Lotka-Volterra model for finite signed graphs] \label{thm:lv-sgn-graph}
  There is a functor
  \begin{equation*}
    \LV: \FinSgnGraph \to \ConParaDynam
  \end{equation*}
  that sends a finite signed graph $X$ to the conically parameterized
  nonnegative dynamical system with parameter variables
  $P \coloneqq X(V) + X(E)$, state variables $S \coloneqq X(V)$, and essentially
  nonnegative, algebraic vector field
  \begin{equation*}
    v(x; \rho, \beta)(i) \coloneqq -\rho(i)\, x(i) +
      \sum_{(e: i' \to i) \in X} X(\sgn)(e)\, \beta(e)\, x(i')\, x(i),
    \qquad x \in \R^{X(V)},\ i \in X(V),
  \end{equation*}
  parameterized by $\rho \in \R_+^{X(V)}$ and $\beta \in \R_+^{X(E)}$. Moreover,
  the functor $\LV$ preserves finite colimits.
\end{theorem}
\begin{proof}
  Similarly to the previous proof, the functor
  $\LV: \FinSgnGraph \to \ConParaDynam$ is defined by functors
  $\LV_0, \LV_1: \FinSgnGraph \to \FinSet$ along with a natural transformation
  \begin{equation*}
    \vec\LV: F_+ \circ \LV_0 \To \Dynam_+ \circ \LV_1: \FinSgnGraph \to \Con,
  \end{equation*}
  now having components $\vec\LV_X$ given by the copairing of
  \begin{equation*}
    v_X^0 \coloneqq \vec\LV_X^0: \R_+^{X(V)} \to \Dynam_+(X(V))
    \quad\text{and}\quad
    v_X^1 \coloneqq \vec\LV_X^1: \R_+^{X(E)} \to \Dynam_+(X(V)),
  \end{equation*}
  where we define
  \begin{equation*}
    v_X^0(x; \rho)(i) \coloneqq -\rho(i)\, x(i)
    \quad\text{and}\quad
    v_X^1(x; \beta)(i) \coloneqq \sum_{e \in X(\tgt)^{-1}(i)}
      X(\sgn)(e)\, \beta(e)\, x(X(\src)(e))\, x(i).
  \end{equation*}
  The proof of naturality is essentially the same as before, using the crucial
  additional fact that morphisms of signed graphs preserve signs. The proof that
  the functor $\LV$ preserves finite colimits is unchanged.
\end{proof}

To exemplify the theorem, let us see how the Lotka-Volterra dynamics functor
acts on a monomorphism and on an epimorphism of signed graphs.

In order to compare the dynamics of two species $A$ and $B$ involved in a
negative feedback loop versus $A$ and $B$ in isolation, we take the inclusion of
signed graphs
\begin{equation*}
  \begin{tikzcd}[column sep=small]
    A & B & {} & {} & A & B
    \arrow[curve={height=-12pt}, from=1-5, to=1-6]
    \arrow[curve={height=12pt}, maps to, no head, from=1-5, to=1-6]
    \arrow["\iota", hook, from=1-3, to=1-4]
  \end{tikzcd}
\end{equation*}
Labeling the edges in the feedback loop as $AB$ and $BA$, the morphism
$\LV(\iota)$ sends the conically parameterized dynamical system
\begin{equation*}
  \left\{
    \begin{aligned}
      v_A(x; \rho) &= -\rho_A\, x_A \\
      v_B(x; \rho) &= -\rho_B\, x_B
    \end{aligned}
  \right.,
  \qquad
  \rho \in \R_+^{\{A,B\}},
\end{equation*}
to the parameterized dynamical system
\begin{equation*}
  \left\{
    \begin{aligned}
      v_A(x; \rho, \beta) &= -\rho_A\, x_A - \beta_{BA}\, x_B\, x_A \\
      v_B(x; \rho, \beta) &= -\rho_B\, x_B + \beta_{AB}\, x_A\, x_B
    \end{aligned}
  \right.,
  \qquad
  \rho \in \R_+^{\{A,B\}},\ \beta \in \R_+^{\{AB,BA\}},
\end{equation*}
by setting the latter's interaction coefficients to zero:
$\beta_{AB} = \beta_{BA} = 0$. This formalizes the commonsense fact that the
first system is a degenerate case of the second.

For a more interesting example, we return to the projection map between
regulatory networks given by \cref{eq:arginine-morphism} of
\cref{sec:sgn-graph}, inspired by the arginine biosynthesis system. Call this
projection map $p$, and abbreviate the regulator molecule as $R$ and the enzymes
as $S \coloneqq \{C,D,E,F,I\}$. The morphism $\LV(p)$ sends the parameterized
dynamical system
\begin{equation*}
  \left\{
    \begin{aligned}
      v_R(x; \rho, \beta) &= -\rho_R\, x_R - \beta_R\, x_R^2 \\
      v_C(x; \rho, \beta) &= -\rho_C\, x_C - \beta_C\, x_R\, x_C \\
      v_D(x; \rho, \beta) &= -\rho_D\, x_D - \beta_D\, x_R\, x_D
    \end{aligned}
    \qquad\qquad
    \begin{aligned}
      v_E(x; \rho, \beta) &= -\rho_E\, x_E - \beta_E\, x_R\, x_E \\
      v_F(x; \rho, \beta) &= -\rho_F\, x_F - \beta_F\, x_R\, x_F \\
      v_I(x; \rho, \beta) &= -\rho_I\, x_I - \beta_I\, x_R\, x_I
    \end{aligned}
  \right.
\end{equation*}
with state variables $\{R\} + S$ and parameters $\rho, \beta \in \R_+^{\{R\}+S}$
to the parameterized dynamical system
\begin{equation*}
  \left\{
    \begin{aligned}
      v_R(x; \rho, \beta) &= -\rho_R\, x_R - \beta_R\, x_R^2 \\
      v_*(x; \rho, \beta) &= -\rho_*\, x_* - \beta_*\, x_R\, x_*
    \end{aligned}
  \right.
\end{equation*}
with state variables $\{R,*\}$ and parameters $\rho, \beta \in \R_+^{\{R,*\}}$,
in two different but equivalent ways. The first way sets the latter system's
coefficients equal to sums of the former's coefficients, namely
\begin{equation*}
  \rho_* = \sum_{i \in S} \rho_i
  \qquad\text{and}\qquad
  \beta_* = \sum_{i \in S} \beta_i.
\end{equation*}
The second way substitutes $x_*$ for each $x_i$, $i \in S$, in the first system
and then takes the vector field $v_*$ to be the sum of the $v_i$'s, $i \in S$,
with these substitutions. The equivalence of these operations is precisely the
condition for $\LV(p)$ to be a morphism of parameterized dynamical systems, cf.\
\cref{eq:para-dynam-hom,eq:lv-naturality}.

\subsection{Composing Lotka-Volterra models}
\label{sec:open-lv}

To complete this part of the story, we extend the Lotka-Volterra dynamics
functors between graphs and parameterized dynamical systems, constructed in
\cref{thm:lv-graph,thm:lv-sgn-graph}, to \emph{double} functors between open
graphs and open parameterized dynamical systems. We begin by making
parameterized dynamical systems into open systems. By ``open systems,'' we mean
dynamical systems that have specified interfaces along which they can share
material with other systems.

\begin{proposition}[Open parameterized dynamical systems]
  There is a symmetric monoidal double category of \emph{open} linearly
  parameterized dynamical systems, $\Open{\LinParaDynam}$, having
  \begin{itemize}[noitemsep]
    \item as objects, finite sets $A, A', \dots$;
    \item as vertical morphisms, functions $f: A \to A'$;
    \item as horizontal morphisms, \define{open linearly parameterized dynamical
      systems}, which consist of a linearly parameterized dynamical system
      $(P, S, v: \R^P \to \Dynam(S))$ along with a cospan
      $A_0 \xrightarrow{\ell_0} S \xleftarrow{\ell_1} A_1$ whose apex is the set $S$
      of state variables;
    \item as cells, \define{morphisms of such open systems}
      $(P, S, v, \ell_0, \ell_1) \to (P', S', v', \ell_0', \ell_1')$, which
      consist of a morphism $(q,f): (P,S,v) \to (P',S',v')$ between linearly
      parameterized dynamical systems along with functions $f_0: A_0 \to A_0'$
      and $f_1: A_1 \to A_1'$ making the diagram commute:
      \begin{equation*}
        \begin{tikzcd}
          {A_0} & S & {A_1} \\
          {A_0'} & {S'} & {A_1'}
          \arrow["{\ell_0}", from=1-1, to=1-2]
          \arrow["{\ell_1}"', from=1-3, to=1-2]
          \arrow["{f_0}"', from=1-1, to=2-1]
          \arrow["f"', from=1-2, to=2-2]
          \arrow["{\ell_0'}"', from=2-1, to=2-2]
          \arrow["{f_1}", from=1-3, to=2-3]
          \arrow["{\ell_1'}", from=2-3, to=2-2]
        \end{tikzcd}.
      \end{equation*}
  \end{itemize}
  Vertical composition is by composition in $\FinSet$ and in $\LinParaDynam$.
  Horizontal composition and monoidal products are by pushouts and coproducts in
  $\LinParaDynam$, respectively, interpreting the finite sets in the feet of the
  cospans as linearly parameterized dynamical systems with no parameter
  variables and identically zero vector fields.

  Similarly, there is a symmetric monoidal double category
  $\Open{\ConParaDynam}$ of \emph{open} conically parameterized nonnegative
  dynamical systems.
\end{proposition}
\begin{proof}
  We do the construction for linearly parameterized dynamical systems. The
  construction for conically parameterized nonnegative dynamical systems is
  perfectly analogous, replacing $\R$ with $\R_+$ and vector spaces with conical
  spaces.

  The projection functor $\pi_S: \LinParaDynam \to \FinSet$, $(P,S,v) \mapsto S$
  that sends a linearly parameterized dynamical systems to its set of state
  variables has a left adjoint $Z: \FinSet \to \LinParaDynam$ that sends a
  finite set $S$ to the system $(\emptyset, S, 0)$ with empty set of parameter
  variables. By linearity, its parameterized vector field
  $0 \cong \R^\emptyset \to \Dynam(S)$ is necessarily the zero vector field.
  This indeed gives an adjunction $Z \dashv \pi_S$, because to any function
  $f: S \to S'$ and linearly parameterized dynamical system $(P',S',v')$ there
  corresponds a unique morphism $(0_{P'}, f): Z(S) \to (P',S',v')$, where the
  compatibility square
  \begin{equation*}
    \begin{tikzcd}
      0 & {\Dynam(S)} \\
      {\R^{P'}} & {\Dynam(S')}
      \arrow[from=1-1, to=1-2]
      \arrow["{f_* \circ (-) \circ f^*}", from=1-2, to=2-2]
      \arrow["{v'}"', from=2-1, to=2-2]
      \arrow[from=1-1, to=2-1]
    \end{tikzcd}
  \end{equation*}
  commutes trivially, since the zero vector space is initial in $\Vect$.

  Therefore, since $\LinParaDynam$ has finite colimits
  (\cref{prop:colim-para-dynam}), we can construct a symmetric monoidal double
  category of $Z$-structured cospans \cite[Theorem 3.9]{baez2020}, which is
  isomorphic to $\Open{\LinParaDynam}$ by arguments given before.
\end{proof}

We can now construct double functors between open graphs and open parameterized
dynamical systems, but the vertex parameters under Lotka-Volterra dynamics cause
a twist in the story compared to Baez and Pollard's compositionality result for
mass-action kinetics of open Petri nets \cite[\mbox{Theorem 18}]{baez2017}. When
composing open dynamical systems in the image of the Lotka-Volterra functor, one
takes a coproduct of the parameter variables, i.e., a direct sum of the
parameter spaces, belonging to identified vertices. However, if one composes the
open graphs first, then the identified vertices receive a single copy of the
parameters from the Lotka-Volterra functor. Thus this functor does not preserve
composition of open systems, not even up to isomorphism. Nevertheless, there is
a (noninvertible) comparison between the two: given a pair of vertex parameters
in the direct sum, we can reduce them to a single parameter simply by summing
them. In mathematical terms, we get a \emph{lax} double functor: a double
functor that strictly preserves vertical composition, as usual, but preserves
horizontal composition only up to specified comparison cells.\footnote{The
  precise definition of a lax double functor can be found, for instance, in the
  textbook \cite[\S{3.5}]{grandis2019}.} While this laxness could be seen as a
failure of compositionality, it is at most a very mild and well controlled
failure. It is better regarded as a bookkeeping device for the vertex
parameters.

\begin{theorem}[Open Lotka-Volterra models] \label{thm:open-lv}
  There is a symmetric monoidal \emph{lax} double functor
  \begin{equation*}
    \LV: \Open{\FinGraph} \to \Open{\LinParaDynam}
  \end{equation*}
  that acts
  \begin{itemize}[noitemsep]
    \item on objects and vertical morphisms, as the identity;
    \item on horizontal morphisms and cells, by the functor
      $\LV: \FinGraph \to \LinParaDynam$ on graphs and graph homomorphisms and as
      the identity on the associated cospans and cospan morphisms:
      \begin{equation*}
        \Big(X,\, A_0 \xrightarrow{\ell_0} X(V) \xleftarrow{\ell_1} A_1\Big) \mapsto
        \Big(\LV(X),\, A_0 \xrightarrow{\ell_0} X(V) \xleftarrow{\ell_1} A_1\Big).
      \end{equation*}
  \end{itemize}
  The comparison cells are defined using the morphisms of linearly parameterized
  dynamical systems $\alpha_S: Z(S) \to \LV(\Disc S)$, where
  \begin{equation*}
    \alpha_S \coloneqq (0_S, 1_S): (\emptyset, S, 0) \to (S, S, \vec\LV(\Disc S)),
    \qquad S \in \FinSet.
  \end{equation*}
  \begin{itemize}[noitemsep]
    \item Given composable open graphs $(X,\, A \rightarrow X(V) \leftarrow B)$
      and $(Y,\, B \rightarrow Y(V) \leftarrow C)$, the comparison cell for
      horizontal composition is given by the morphism of systems
      \begin{equation*}
        \LV(X) +_{Z(B)} \LV(Y)
          \xrightarrow{\id +_{\alpha_B} \id} \LV(X) +_{\LV(\Disc B)} \LV(Y)
          \xrightarrow{\cong} \LV(X +_{\Disc B} Y).
      \end{equation*}
    \item Given a finite set $A$, the comparison cell for the horizontal unit is
      given by the morphism of systems $\alpha_A: Z(A) \to \LV(\Disc A)$.
  \end{itemize}
  Similarly, there is a symmetric monoidal \emph{lax} double functor
  \begin{equation*}
    \LV: \Open{\FinSgnGraph} \to \Open{\ConParaDynam}.
  \end{equation*}
\end{theorem}
\begin{proof}
  To construct the lax double functor, we use \cite[\mbox{Theorem
    2.4}]{patterson2023}, which extends the construction of a double functor in
  \cite[\mbox{Theorem 4.3}]{baez2020} from the pseudo to the lax case. The
  family of morphisms $\alpha_S: Z(S) \to \LV(\Disc S)$, $S \in \FinSet$, in the theorem
  statement assemble into a natural transformation
  \begin{equation*}
    \begin{tikzcd}
      \FinSet & \FinGraph \\
      \FinSet & \LinParaDynam
      \arrow["{L = \Disc}", from=1-1, to=1-2]
      \arrow["\LV", from=1-2, to=2-2]
      \arrow[Rightarrow, no head, from=1-1, to=2-1]
      \arrow["{L' = Z}"', from=2-1, to=2-2]
      \arrow["\alpha", shorten <=9pt, shorten >=9pt, Rightarrow, from=2-1, to=1-2]
    \end{tikzcd}.
  \end{equation*}
  The functors involved in this cell all preserve finite colimits: the top and
  bottom ones because they are left adjoints and the right one by
  \cref{thm:lv-graph}. By \cite[\mbox{Theorem 2.4}]{patterson2023}, we obtain a
  symmetric monoidal lax double functor
  \begin{equation*}
    \Open{\FinGraph} \cong {}_{L}\Csp(\FinGraph) \to
      {}_{L'}\Csp(\LinParaDynam) \cong \Open{\LinParaDynam}.
  \end{equation*}

  To see that this double functor is the same one in the theorem statement, we
  once again use the adjunctions to pass between $L$-structured and
  $R$-decorated cospans (recalling terminology introduced in the proof of
  \cref{prop:open-sgn-graphs}). Notice that the natural transformation $\alpha$ has
  as its mate \cite[\S 1]{cheng2014} the identity transformation
  $\bar\alpha = 1_{\ev_V}$:
  \begin{equation*}
    \begin{tikzcd}
      \FinSet & \FinGraph \\
      \FinSet & \LinParaDynam
      \arrow["{R = \ev_V}"', from=1-2, to=1-1]
      \arrow[Rightarrow, no head, from=1-1, to=2-1]
      \arrow["\LV", from=1-2, to=2-2]
      \arrow["{R' = \pi_S}", from=2-2, to=2-1]
      \arrow["\bar\alpha", shorten <=9pt, shorten >=9pt, Rightarrow, from=1-1, to=2-2]
    \end{tikzcd}.
  \end{equation*}
  Thus the action of the double functor $F \coloneqq \LV$ on $L$-structured
  cospans simplifies to the identity when translated to $R$-decorated cospans.
  \begin{equation*}
    \begin{tikzcd}[column sep=small]
      & {L(A_0)} & X & {L(A_1)} & \\
      {L'(A_0)} & {F(L(A_0))} & {F(X)} & {F(L(A_1))} & {L'(A_1)}\\
      & & \rotatebox[origin=c]{90}{$\leftrightsquigarrow$} & & \\
      & {A_0} & {R(X)} & {A_1} & \\
      {A_0} & {R(X)} & {R'(F(X))} & {R(X)} & {A_1}
      \arrow["{\ell_0}", from=1-2, to=1-3]
      \arrow["{\ell_1}"', from=1-4, to=1-3]
      \arrow["{\alpha_{A_0}}", from=2-1, to=2-2]
      \arrow["{F(\ell_0)}", from=2-2, to=2-3]
      \arrow[shorten <=3pt, shorten >=3pt, maps to, from=1-3, to=2-3]
      \arrow["{F(\ell_1)}"', from=2-4, to=2-3]
      \arrow["{\alpha_{A_1}}"', from=2-5, to=2-4]
      \arrow["{\bar\ell_0}", from=4-2, to=4-3]
      \arrow["{\bar\ell_1}"', from=4-4, to=4-3]
      \arrow["{\bar\ell_0}", from=5-1, to=5-2]
      \arrow["{\bar\alpha_X}", from=5-2, to=5-3]
      \arrow[shorten <=3pt, shorten >=3pt, Rightarrow, no head, from=4-3, to=5-3]
      \arrow["{\bar\alpha_X}"', from=5-4, to=5-3]
      \arrow["{\bar\ell_1}"', from=5-5, to=5-4]
    \end{tikzcd}
  \end{equation*}
  A similar statement holds for the action of the double functor on morphisms of
  $L$-structured and $R$-decorated cospans.
\end{proof}

\section{Conclusion}

\paragraph{Summary.}

Regulatory networks are a minimalistic but widely used tool to describe the
interactions between molecules in biochemical systems. We have made the first
functorial study of regulatory networks, formalized as signed graphs, and their
connections with other mathematical models in biochemistry. Among such models,
we have studied reaction networks, formalized as Petri nets with signed links,
and parameterized dynamical systems, focusing on Lotka-Volterra dynamics.

The major categories of this paper, and the functors between them, are
summarized in the following diagram, where ``LV'' is the Lotka-Volterra dynamics
functor (\S\ref{sec:lv}).
\begin{equation*}
  \begin{tikzcd}
    \SgnCat\ (\S\ref{sec:sgn-cat}) & \SgnGraph\ (\S\ref{sec:sgn-graph}) & \SgnPetri\ (\S\ref{sec:sgn-petri}) \\
    & \FinSgnGraph & \ConParaDynam\ (\S\ref{sec:para-dynam})
    \arrow[hook', from=2-2, to=1-2]
    \arrow[""{name=0, anchor=center, inner sep=0}, "\Path"', curve={height=18pt}, from=1-2, to=1-1]
    \arrow[""{name=1, anchor=center, inner sep=0}, "U"', curve={height=18pt}, from=1-1, to=1-2]
    \arrow["\LV", from=2-2, to=2-3]
    \arrow["\Int"', from=1-3, to=1-2]
    \arrow["\dashv"{anchor=center, rotate=-90}, draw=none, from=0, to=1]
  \end{tikzcd}
\end{equation*}
Most of the main results extend from closed systems to open systems, which
compose by gluing along their boundaries. Of the diagram above, we have extended
the following parts to double categories of open systems and double functors
between them. The Lotka-Volterra double functor is lax; the others are pseudo.
\begin{equation*}
  \begin{tikzcd}
    {\Open{\SgnCat}} & {\Open{\SgnGraph}} \\
    & {\Open{\FinSgnGraph}} & {\Open{\ConParaDynam}\ (\S\ref{sec:open-lv})}
    \arrow[hook', from=2-2, to=1-2]
    \arrow["\LV", from=2-2, to=2-3]
    \arrow["\Path"', from=1-2, to=1-1]
  \end{tikzcd}
\end{equation*}

\paragraph{Outlook.}

Of many possible directions for future work, we mention a few. As noted in the
introduction, Lotka-Volterra dynamics are only one of numerous dynamics that
could be considered as a canonical model for regulatory networks, and they are
not even among the most commonly studied in the biochemistry literature
\cite{tyson2019}. It would be desirable to have dynamics functors for regulatory
networks that draw on more flexible or more biologically plausible classes of
dynamical systems. In another direction, the two halves of this
paper---qualitative and quantitative---are not as tightly as integrated as one
might hope. How does the presence of a motif in a regulatory network, such as an
incoherent feedforward loop perhaps even of a specific type, manifest in the
continuous dynamics of that network? Put in category-theoretic terms, the
Lotka-Volterra dynamics functor is defined on signed graphs; how does it relate
to the freely generated signed categories in which motifs are expressed? These
intriguing questions are suggestive of ``feedback loop analysis'' in the field
of system dynamics \cite{richardson1995}, to which stronger connections should
be made.

This project fits into a broader program by applied category theorists and other
scientists that aims to systematize, in a completely precise way, the language
and methods of describing, comparing, and composing scientific models in
different domains. Within biology, the field of systems biology has advocated
for a holistic view of complex biological systems that emphasizes composition as
much as reduction. We believe that category theory has a role to play in this
endeavor by bringing mathematical precision to compositional and structural
aspects of modeling that are traditionally thought to be outside the realm of
mathematics.

\sloppy
\printbibliography[heading=bibintoc]

\end{document}